%% file: manuscript_final.tex
\documentclass[review]{elsarticle}

\usepackage{hyperref}

\journal{Stochastic Processes and Their Applications}

\usepackage[utf8]{inputenc} 
\usepackage[T1]{fontenc}    
\usepackage{hyperref}       
\usepackage{url}            
\usepackage{booktabs}       
\usepackage{amsfonts}       
\usepackage{nicefrac}       
\usepackage{microtype}      
\usepackage{xcolor} 

\usepackage{amsmath}
\usepackage{amsthm}
\usepackage{tikz}
\usepackage{caption,subcaption}
\usepackage{array}
\usepackage{mdwmath}
\usepackage{multirow}
\usepackage{mdwtab}
\usepackage{eqparbox}
\usepackage{multicol}
\usepackage{amsfonts}
\usepackage{tikz}
\usepackage{multirow,bigstrut,threeparttable}
\usepackage{amsthm}
\usepackage{array}
\usepackage{bbm}
\usepackage{epstopdf}
\usepackage{mdwmath}
\usepackage{mdwtab}
\usepackage{eqparbox}
\usepackage{tikz}
\usepackage{latexsym}
\usepackage{amssymb}
\usepackage{bm}
\usepackage{amssymb}
\usepackage{graphicx}
\usepackage{mathrsfs}
\usepackage{epsfig}
\usepackage{psfrag}
\usepackage{setspace}
\usepackage{hyperref}
\usepackage{url}
\usepackage[ruled]{algorithm2e}
\usepackage{algpseudocode}
\usepackage{mathtools}
\usepackage{comment}
\usepackage{enumitem, kantlipsum}
\usepackage{mleftright}
\usepackage{floatrow}

\newfloatcommand{capbtabbox}{table}[][\FBwidth]

\theoremstyle{plain}
\newtheorem{theorem}{Theorem}

\newtheorem{lemma}{Lemma}

\newtheorem{corollary}{Corollary}

\newtheorem{assumption}{Assumption}

\def \var {\mathsf{Var}}

\def\EE{\mathbb{E}}

\def\PP{\mathbb{P}}

\def\RR{\mathbb{R}}

\def\bX{\mathbf{X}}

\usepackage{xspace}

\definecolor{myblue}{rgb}{.8, .8, 1}
\definecolor{mathblue}{rgb}{0.2472, 0.24, 0.6} 
\definecolor{mathred}{rgb}{0.6, 0.24, 0.442893}
\definecolor{mathyellow}{rgb}{0.6, 0.547014, 0.24}

\newcommand{\calB}{{\mathcal{B}}}

\newcommand{\calF}{{\mathcal{F}}}

\newcommand{\calN}{{\mathcal{N}}}
\newcommand{\calO}{{\mathcal{O}}}

\newcommand{\calS}{{\mathcal{S}}}
\newcommand{\calT}{{\mathcal{T}}}

\newcommand{\calX}{{\mathcal{X}}}

\makeatletter
\newcommand*\rel@kern[1]{\kern#1\dimexpr\macc@kerna}
\newcommand*\widebar[1]{%
  \begingroup
  \def\mathaccent##1##2{%
    \rel@kern{0.8}%
    \overline{\rel@kern{-0.8}\macc@nucleus\rel@kern{0.2}}%
    \rel@kern{-0.2}%
  }%
  \macc@depth\@ne
  \let\math@bgroup\@empty \let\math@egroup\macc@set@skewchar
  \mathsurround\z@ \frozen@everymath{\mathgroup\macc@group\relax}%
  \macc@set@skewchar\relax
  \let\mathaccentV\macc@nested@a
  \macc@nested@a\relax111{#1}%
  \endgroup
}

\makeatother

\begin{document}

\begin{frontmatter}

\title{Unbiased Optimal Stopping via the MUSE}

\author{Zhengqing Zhou \fnref{myfootnote}}
\address{Department of Mathematics, Stanford University}
\fntext[myfootnote]{Zhengqing Zhou and Guanyang Wang have contributed equally to this paper.}
\author{Guanyang Wang \fnref{myfootnote}}
\address{Department of Statistics, Rutgers University}
\author{Jose H. Blanchet, Peter W. Glynn}
\address{Management Science and Engineering, Stanford University}



\begin{abstract}
 We propose a new unbiased estimator for estimating the utility of the optimal stopping problem. The MUSE, short for ‘Multilevel Unbiased Stopping Estimator’, constructs the unbiased Multilevel Monte Carlo (MLMC) estimator at every stage of the optimal stopping problem in a backward recursive way. In contrast to traditional sequential methods, the MUSE can be implemented in parallel. We prove the MUSE has finite variance, finite computational complexity, and achieves $\varepsilon$-accuracy with $O(1/\varepsilon^2)$ computational cost under mild conditions. We demonstrate MUSE empirically in an option pricing problem involving a high-dimensional input and the use of many parallel processors.
\end{abstract}


\end{frontmatter}

\section{Introduction}\label{sect.intro}
It is a pleasure to contribute to this special issue in honor of Prof. Larry Shepp, whose work has significantly impacted a wide range of scientific disciplines. This paper focuses on optimal stopping problems, an area of stochastic control in which Prof. Shepp contributed deeply both in terms of theory and applications. in particular in connection to mathematical finance problems; see, for example, \cite{Shepp96, Shepp97, Shepp93a, Shepp93b}. 

Our goal in this paper is on designing Monte Carlo methods for solving optimal stopping problems which can be easily parallelized because the estimators that we produce are unbiased, are applicable even in non-Markovian problems and have finite variance. We are not aware of any other Monte Carlo estimators for optimal stopping problems which share these properties. 

Monte Carlo methods are ubiquitously used for estimating high dimensional numerical integrals or statistics arising in every computation-related subject. However, vanilla Monte Carlo estimators may produce systematic bias in many practical applications (in particular those involving optimization). The presence of such systematic bias  precludes the direct use of parallel computing architectures. Consider the following toy example in a two-stage optimal stopping problem. Suppose one is able to simulate the two-stage process $(X_1,X_2)$ and is interested in estimating the utility
\begin{equation*}
    U := \EE \left[\max\lbrace f(X_1), \EE\left[f(X_2)\mid X_1\right] \rbrace\right],
\end{equation*}
where $f$ is some integrable reward function. Then it is not hard to show the vanilla Monte Carlo estimator $\widehat U$ will systematically overestimate $U$ by the Jensen's inequality.

To address the issue of bias, the design of unbiased Monte Carlo estimators has recently attracted much attention \cite{mcleish2011general, glynn2014exact, rhee2015unbiased, Blanchet2015UnbiasedMC,vihola2018unbiased, Blanchet2019Unbiased,biswas2019estimating,heng2019unbiased,jacob2020unbiased,middleton2020unbiased, wang2021maximal,wang2022unbiased} in operations research, statistics, and machine learning communities. Many existing debiasing techniques are closely  related to the Multilevel Monte Carlo (MLMC) framework developed by  Heinrich and Giles \cite{heinrich2001multilevel, giles2008multilevel,giles2015multilevel, giles2014antithetic, giles2019decision} where a sequence of biased but increasingly accurate estimators are used to estimate the functionals of stochastic processes  described by stochastic differential equations. Unbiased estimators have been successfully developed and employed in the context of stochastic approximation, Markov chain Monte Carlo estimation and convergence diagnosis, quantile estimation, and so on. In many settings, unbiased estimators provide a promising direction to efficient parallel implementation and better uncertainty quantification. 

In this paper, we study the discrete-time, finite-horizon optimal stopping problem, a fundamental problem  that can be found in areas including economics, operations research, and financial engineering. Consider the optimal stopping problem with underlying process $(X_1,\cdots, X_T)$ and reward function $f$.  We are interested in computing the expected utility of the optimal strategy:
\begin{equation}\label{def.optimal_stopping}
    U_T :=\sup_{\tau\in\calT_{T}} \EE \left[f\left(X_\tau\right)\right],
\end{equation}
where $\calT_{T}$ denotes the set of all the stopping times taking values in $\{1,\cdots, T\}$. Following the standard optimal stopping theory, we can define the Snell envelope by
\begin{equation*}
    U_{T-k} := \sup_{\tau\in\calT_{k+1, T}} \EE \left[f\left(X_{\tau}\right) \mid \calF_k\right],
\end{equation*}
where $\calT_{k+1, T}$ denotes the set of stopping times satisfying $k+1\leq \tau\leq T$, $k=0,\cdots, T-1$ and $\calF_k$ is the natural filtration at time $k$. The dynamical programming can be written as:
\begin{equation*}
\begin{cases}
U_1= \EE \left[f\left(X_T\right)\mid \calF_{T-1}\right], & \\
U_{T-k} = \EE \left[\max \left\lbrace f(X_{k+1}), U_{T-(k+1)}\right\rbrace\mid \calF_k\right], & \quad k = 0,\cdots, T-1. 
\end{cases}
\end{equation*}

    In most practical cases, $U_T$ cannot be analytically solved, and we therefore resort to Monte Carlo methods for an estimation. However, generating unbiased estimation for the utility of the optimal problem is known as a very difficult problem. Suppose one is able to simulate the whole process. Then the above dynamical programming backward recursion suggests a natural Monte Carlo estimator as follows. We sample tree-like paths of the whole process forward in time and estimate each $U_i$ backwards in time. The sampling procedure is illustrated in Fig \ref{fig.Tree} where one samples many $k$-ary trees with height $T$. After sampling enough paths, we aggregate the samples from bottom to top in each layer as estimators of $U_1,\cdots, U_T$ respectively. This estimator is relatively easy to implement but has several limitations. Firstly,  the estimator overestimates the utility   even in the simplest case  $T = 2$, let alone the general case. Secondly, as the estimation error propagates from one time horizon to another,  the accuracy  relies on a repeated $T$-limit, which is difficult to quantify. The above approach is the `high-estimator'  suggested in the seminal paper of  Broadie and Glasserman \cite{broadie1997pricing}. The authors also use the similar idea to construct the `low-estimator', and a confidence interval that covers the utility by combining the two estimators.  In fact, the above authors conjectured that there is no general unbiased estimators for the optimal stopping problem, see page 1326-1327 in \cite{broadie1997pricing} for details. There are also regression-based  Monte Carlo simulation methods, including the well-celebrated  Longstaff–Schwartz \cite{longstaff2001valuing} and  Tsitsiklis–Van Roy \cite{tsitsiklis2001regression} algorithms for option pricing, see also \cite{egloff2005monte} for extensions. Both methods approximate the solution of the original problem by solving a sequence of regression problems in linearly parameterized subspaces. Albeit convenient to use, the approximation error  and the unavoidable bias still cause concerns for both parallel implementation and uncertainty quantification. 

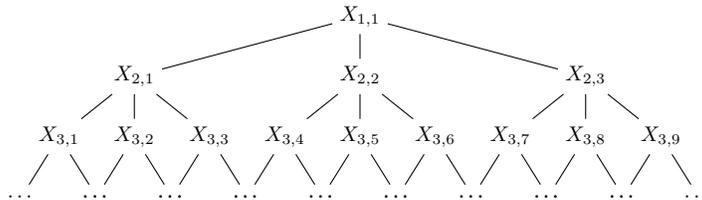
\begin{figure}[htbp]
    \centering
\begin{tikzpicture}[level distance= 0.8cm,
  level 1/.style={sibling distance=3cm},
  level 2/.style={sibling distance=1cm}
]
  \node [draw=none, scale = 0.8] {$X_{1,1}$}
    child {node [draw = none, scale = 0.8] {$X_{2,1}$}
      child {node [draw = none, scale = 0.8] {$X_{3,1}$}
      child {node [draw = none, scale = 0.8] {$\cdots$}}
      child {node [draw = none, scale = 0.8] {$\cdots$}}
      }
      child {node [draw = none, scale = 0.8] {$X_{3,2}$}
      child {node [draw = none, scale = 0.8] {$\cdots$}}
      child {node [draw = none, scale = 0.8] {$\cdots$}}}
      child {node [draw = none, scale = 0.8] {$X_{3,3}$}
      child {node [draw = none, scale = 0.8] {$\cdots$}}
      child {node [draw = none, scale = 0.8] {$\cdots$}}}
    }
    child {node[draw = none, scale = 0.8] {$X_{2,2}$}
    child {node [draw = none, scale = 0.8]{$X_{3,4}$}
    child {node [draw = none, scale = 0.8] {$\cdots$}}
    child {node [draw = none, scale = 0.8] {$\cdots$}}}
      child {node [draw = none, scale = 0.8] {$X_{3,5}$}
      child {node [draw = none, scale = 0.8] {$\cdots$}}
      child {node [draw = none, scale = 0.8] {$\cdots$}}}
      child {node [draw = none, scale = 0.8] {$X_{3,6}$}
      child {node [draw = none, scale = 0.8] {$\cdots$}}
      child {node [draw = none, scale = 0.8] {$\cdots$}}}
    }
    child {node[draw = none, scale = 0.8] {$X_{2,3}$}
    child {node [draw = none,scale = 0.8]{$X_{3,7}$}
    child {node [draw = none, scale = 0.8] {$\cdots$}}
    child {node [draw = none, scale = 0.8] {$\cdots$}}}
      child {node [draw = none, scale = 0.8] {$X_{3,8}$}
      child {node [draw = none, scale = 0.8] {$\cdots$}}
      child {node [draw = none, scale = 0.8] {$\cdots$}}}
      child {node [draw = none, scale = 0.8] {$X_{3,9}$}
      child {node [draw = none, scale = 0.8] {$\cdots$}}
      child {node [draw = none, scale = 0.8] {$\cdots$}}}
    };
\end{tikzpicture}
    \caption{Tree-like paths for Monte Carlo simulation. Here each node in the first two levels has three children which are $i.i.d.$ sampled from the conditional distribution.}
    \label{fig.Tree}
\end{figure}
In this paper, we introduce a novel unbiased estimator -- the Multilevel Unbiased Stopping Estimator (MUSE \footnote{In ancient Greek mythology, the Muses are the nine goddesses (the daughters of Zeus and Mnemosyne), who preside over literature, science, and the arts.}) for the optimal stopping problem described in \eqref{def.optimal_stopping}. Our estimator is inspired by the randomized Multilevel Monte Carlo estimator described in \cite{Blanchet2015UnbiasedMC, Blanchet2019Unbiased}. The MUSE is easy to implement and enjoys  both finite variance and finite expected computational complexity. The computational cost to achieve $\varepsilon$-accuracy is  $O(1/\varepsilon^2)$, which matches the optimal rate from the Central Limit Theorem (CLT). Our empirical studies suggest that the MUSE scales well with the dimensionality of the underlying process -- which is often viewed as a bottleneck of classical regression-based methods. As extra byproducts, we construct confidence intervals for the utility and  propose a natural algorithm to determine the optimal stopping time  based on the MUSE.

We emphasize that our techniques are of interest in randomized Multilevel Monte Carlo as we relax smoothness assumptions imposed in  \cite{Blanchet2015UnbiasedMC}. In summary, the MUSE can be viewed as a multi-stage extension of the MLMC estimator. In two-stage problems, our estimator (Algorithm \ref{alg.one_step_MLMC}) has the same expression as the randomized MLMC estimator. However, the general optimal stopping problem is defined in a recursive way, and thus the MLMC estimator fails to  directly apply. Multi-stage MUSE (Algorithm \ref{alg.multi_stage_MLMC}) generates data forwardly and then calls for the two-stage MUSE backwardly.  Moreover, MUSE relaxes the technical assumptions of the randomized MLMC estimator. In \cite{Blanchet2015UnbiasedMC}, the randomized MLMC estimator is proposed to estimate $g(m_\mu)$ where  $m_\mu$ is the mean of a probability measure $\mu$ and $g$ is a locally twice differentiable function  (Section $3$, Assumption $2$ in \cite{Blanchet2015UnbiasedMC}), but the function of our interest is a multiple composition of the $\max$ function, which is not  differentiable everywhere.

There is a large body of literature on solving the optimal stopping problem using the (non-randomized) MLMC methods. For the two-stage optimal stopping problem, a slightly more general version has been considered in \cite{giles2019decision}. The authors in \cite{giles2019decision} consider the problem of estimating:
\begin{align}\label{eqn: EVPI-EVPPI}
    \EE[\max_d f_d(X,Y)] - \EE[\max_d \EE[f_d(X,Y)\mid X]] 
\end{align}
where $\{f_d\}$ is a finite set of functions representing one's possible strategies. The quantity \eqref{eqn: EVPI-EVPPI} represents the expected value of the partial information. The first term in \eqref{eqn: EVPI-EVPPI} can be estimated via standard Monte Carlo approach, while the second term is much more complicated. Suppose $d = 2, f_1(X,Y) = f(X)$ and $f_2(X,Y) = \EE[f(Y)\mid X]$ for the utility function $f$ given at the beginning of our paper, then the second term in \eqref{eqn: EVPI-EVPPI} is the expected utility of the two-stage optimal stopping problem. The methodology and theory of \cite{giles2019decision} are closely connected to our paper, though the focus is somewhat different.  The main contribution of \cite{giles2019decision} is applying the (non-randomized) MLMC method to design low-bias 
estimators which attains $\epsilon^2$ mean-squared error with  $\calO(\epsilon^2 (\log\epsilon)^2)$ or $\calO(\epsilon^2)$ computational complexity under different assumptions. In contrast, our effort is mostly on designing completely unbiased estimators with finite variance and finite computational cost. Technically, Assumption 2 in \cite{giles2019decision} is very close to our Assumption \ref{ass.conditional} and is posed for similar reasons (see \ref{sect.assumption} for detailed discussions), while our paper has a relatively weaker moment assumption (Assumption \ref{ass.moment}) than Assumption 1 on \cite{giles2018mlmc}. Besides \cite{giles2019decision}, other related works for the two-stage optimal stopping problem include \cite{giles2018mlmc, hironaka2020multilevel, hong2009estimating}, and the references therein. 

Non-randomized MLMC methods have also been used in  the general multi-stage optimal stopping problems to obtain low-bias estimators.  In Belomestny, Ladkau, and Schoenmakers \cite{belomestny2015multilevel}, the authors designed MLMC  estimators to improve the computational complexity of the standard Monte Carlo estimator based on their previous work \cite{belomestny2012tight, belomestny2013multilevel}. The randomized MLMC idea has been briefly mentioned in the Ph.D. dissertation of Dickmann \cite{dickmann2014multilevel} but not explored in details.

The rest of this paper is organized as follows. After setting up the notation and describing the assumptions in Section \ref{sect.assumption}, in Section \ref{sect.main_results} we introduce the MUSE and prove its theoretical properties. In Section \ref{sect.numerical_experiments} we showcase several applications of the MUSE through numerical examples. In Section \ref{sect.conclusion} we conclude our paper with a discussion on its limitations and potential generalizations. Detailed proofs of Theorem \ref{thm.one_step_unbiased} and  \ref{thm.multi_stage_unbiased} are deferred to the Appendix.

\subsection{Notations and assumptions} \label{sect.assumption}
Let $(\Omega,\calB(\Omega))$ be a Polish space equipped with Borel $\sigma$-algebra $\calB(\Omega)$. Let $(\calX, \|\cdot\|)$ be a complete separable normed space equipped with Borel $\sigma$-algebra $\calB(\calX)$.  Given a fixed positive integer $T$ and an adapted $\calX-$ valued stochastic process $\lbrace X_i \rbrace_{i=1}^T$ with filtration $\{\calF_i\}_{i=1}^{T}$, we denote by $\pi_{1:T}$ the joint distribution of $(X_1, \cdots, X_T)$. Given $X_i = x_i$ for $1\leq i \leq k$  , we denote by $\pi_{k+1:T}(\cdot \mid \{x_i\}_{i=1}^k)$ the conditional distribution of $(X_{k+1}, \cdots, X_T)$, and $\pi_{{k+1}}(\cdot \mid \{x_i\}_{i=1}^k)$ the marginal conditional distribution of $X_{k+1}$.
We will use the convention $\{x_0\} = \varnothing$ and therefore $\pi_1$ denotes the (unconditioned) marginal distribution of $X_1$. Let $f:\calX \rightarrow \RR$ be an integrable reward function. We denote by $U_T$ the utility of the optimal stopping problem as described in (\ref{def.optimal_stopping}) and  the conditional utility
\begin{equation}\label{eqn.utility} 
    U_{T-k}(x_1,\cdots, x_k):= \sup_{\tau\in\calT_{k+1, T}} \EE_{\pi_{k+1:T}} \left[f\left(X_\tau\right) \mid \{x_i\}_{i=1}^k \right]
\end{equation}
which corresponds to the utility function if one starts the $T-k$-stage optimal stopping problem after observing the first $k$ outcomes $\{x_i\}_{i=1}^k$. For simplicity, we write $\pi_{k+1:T},\pi_{k+1}$, and $U_{T-k}$ when there is no confusion about the dependency on $\{x_i\}_{i=1}^k$. The geometric distribution taking values in $\{0,1, \cdots\}$ with parameter $r$ is denoted by $\text{Geo}(r)$.  Given a non-negative discrete random variable $N$, we denote its probability mass function by  $p(n):= \PP(N=n)$. In particular, $p_r(n)$ denotes $\PP(\text{Geo}(r)=n) = r(1-r)^{n}$. 


Before formally describing the MUSE, we introduce the following assumption ensuring that the underlying process can be simulated.  This assumption is standard and can be found in almost every Monte Carlo-based optimal stopping algorithm, such as \cite{longstaff2001valuing, tsitsiklis2001regression}.
\begin{assumption}[Path simulation]\label{ass.simulation}
 Given  fixed integers $0\leq k\leq T-1, n \geq 1$ and a trajectory  $\{x_i\}_{i=1}^k$, we have a simulator $\cal S$ which takes $\{x_i\}_{i=1}^k$ as inputs and outputs $n$ $i.i.d.$ samples $X_{k+1}(1), \cdots X_{k+1}(n)$ with distribution $\pi_{k+1}(\cdot \mid\{x_i\}_{i=1}^k)$.  
\end{assumption}
Besides the  simulation assumption, we also introduce several technical assumptions. The MUSE can always be constructed as long as the simulation assumption is satisfied, but  several desired properties such as  finite variance are justified under these technical assumptions.

\begin{assumption}[Moment Assumption]\label{ass.moment}
There exists $\delta > 0$, such that $\EE \left[\|X_i\|^{2+\delta}\right] < \infty$ for all $1\leq i\leq T$. 
\end{assumption}

\begin{assumption}[Linear Growth]\label{ass.lipschitz_f}
$f:\calX\rightarrow \RR$ satisfies $|f(x)|\leq L \left(1 + \|x\|\right)$ for some $L > 0$.
\end{assumption}

\begin{assumption}[Regularity Condition on Conditional Expectation]\label{ass.conditional}
There exists a constant $C > 0$, such that 
\begin{equation*}
    \PP\left( \left| U_{T-k}(X_1,\cdots, X_k) - f\left(X_k\right)\right| < \varepsilon \right) < C\varepsilon
\end{equation*}
hold for all $\varepsilon > 0$ and $1\leq k\leq T-1$. 
\end{assumption}
Assumption \ref{ass.conditional} essentially requires the  density of $U_{T-k}(X_1,\cdots, X_k) - f(X_k)$ to be bounded at zero. This can be verified directly when $(X_1,\cdots, X_T)$ is an independent process and each $X_i$ has a bounded density. In general, we expect Assumption \ref{ass.conditional} to hold if the random vector $(X_1,\cdots, X_T)$ has a density, and the reward function $f$ is smooth except for finitely many points.

To be specific, suppose $(X_1, \cdots, X_k)$ follows the joint distribution $\pi_{1:k}$, then in Assumption 4 we require
\begin{equation*}
\int \mathbbm{1} \left( \left| U_{T-k}(x_1,\cdots, x_k) - f\left(x_k\right)\right| < \varepsilon \right) \pi_{1:k}(dx_1,\cdots, dx_k)\leq C\varepsilon.
\end{equation*}
Note that it is common to assume the underlying information to enjoy certain nice regularity in high-dimensional optimal stopping. For instance, similar regularity assumptions can be found in \cite{belomestny2015multilevel} (Proposition 3.1 and Theorem 3.4 (iv)), and \cite{giles2019decision} (Assumption 2).

We would like to emphasize that the main objective of our technical assumptions here is to bound the variance and expected computational cost simultaneously, which are crucial properties for the efficiency of any unbiased MLMC estimator. The finite variance  enables the Central Limit Theorem, which in turn gives  an asymptotically exact confidence interval of the true utility. The computational cost result allows us to control the expected amount of time for simulating one estimator. Putting the two properties together shows our estimator  achieves $\varepsilon^2$-expected mean squared error within $O(1/\varepsilon^2)$ expected computational cost, see Corollary \ref{cor.application_of_unbiasedness} for details. We also refer the readers to \cite{Blanchet2015UnbiasedMC} and \cite{vihola2018unbiased} for more discussions on the variance, the computational cost and their trade-off for unbiased MLMC methods.

\section{Multilevel Unbiased Stopping Estimator (MUSE)}\label{sect.main_results}
We present our main results in this section. We start with the two-stage MUSE in Section \ref{sect.two-stage muse}. The general/multi-stage MUSE is described in Section \ref{sect.multi-stage MUSE} and it is constructed by recursively calling the two-stage MUSE. Two related applications, including the construction of the confidence interval and an algorithm for finding the optimal stopping time, are discussed in Section \ref{sect.applications}.

\subsection{MUSE for two-stage optimal stopping problems}\label{sect.two-stage muse}
Two-stage optimal stopping is a special and simplest non-trivial case among the finite-horizon optimal stopping problems. To build a better intuition for the MUSE, we start with describing the MUSE under this simplified setting, which serves as both a base case and  motivation for the general estimator. 

Given the bivariate distribution $\pi_{1:2}$ of $(X_1, X_2)$. The utility of the two-stage optimal stopping problem can be written as: 
\begin{equation}\label{eqn.two stage utility}
    U_2 = \EE \left[\max\lbrace f(X_1), \EE\left[f(X_2) \mid X_1\right] \rbrace\right] = \int_{\Omega} \max\{f(x_1), \EE\left[f(X_2)\mid x_1\right] \} \pi_1(d x_1).
\end{equation}
Therefore, after sampling $x_1 \sim \pi_1$ from the simulator $\cal S$, it suffices to construct an unbiased estimator of  $g_{x_1}(\EE\left[f(X_2) \mid x_1\right])$ where $g_{x_1}(a):= \max\{f(x_1), a\}$. It is also clear that a vanilla estimator $\max\{f(x_1), f(x_2)\}$ with $(x_1,x_2)\sim \pi_{1:2}$ will be biased here, as such an estimator has expectation
$$\int_{\Omega \times \Omega} \max\{f(x_1), f(x_2)\}\pi_1(dx_1) \pi_{2}(dx_2\mid x_1) \geq \int_{\Omega} \max\{f(x_1), \EE\left[f(X_2)\mid x_1\right] \} \pi_1(d x_1), 
$$
and therefore overestimates the utility. The debiasing strategy follows from the observation of \cite{Blanchet2015UnbiasedMC, giles2008multilevel}. Let $\{n_0, n_1, \cdots \}$ be an increasing sequence of positive integers. For each $n_i$, let $\widebar{f(X_2)}_{n_i}$ be the empirical average of $n_i$ $i.i.d.$ samples of $f(X_2)$ with $X_2 \sim \pi_{2}(\cdot \mid \{x_1\})$. By virtue of the law of large numbers, we have $g_{x_1}\left(\EE\left[f(X_2) \mid x_1\right]\right)= \lim\limits_{k\rightarrow \infty} g_{x_1}(\widebar{f(X_2)}_{n_k})$ almost surely. Then we can write $g_{x_1}(\EE\left[f(X_2) \mid x_1\right])$  as the following telescoping summation:
\begin{align}\label{eqn.telescoping}
    \nonumber g_{x_1}(\EE\left[f(X_2) \mid x_1\right])&= \lim\limits_{k\rightarrow \infty} g_{x_1}\left(\widebar{f(X_2)}_{n_k}\right)\\
           &=  g_{x_1}\left(\widebar{f(X_2)}_{n_0}\right) + \sum_{k=1}^\infty  g_{x_1}\left(\widebar{f(X_2)}_{n_k}\right) -  g_{x_1}\left(\widebar{f(X_2)}_{n_{k-1}}\right).
\end{align}
If one can construct estimator $\Delta_n$ with expectation 
\begin{align}\label{eqn:delta_n}
    \EE \left[ g_{x_1}\left(\widebar{f(X_2)}_{n_k}\right) -  g_{x_1}\left(\widebar{f(X_2)}_{n_{k-1}}\right)\right] 
\end{align}
for $k\geq 1$ and  $\EE\left[\Delta_0\right] =g_{x_1}\left(\widebar{f(X_2)}_{n_0}\right)$, a randomized estimator of $g_{x_1}(\EE\left[f(X_2) \mid x_1\right])$ can be constructed as $\Delta_N/p_N$, where $N$ is a non-negative integer-valued random variable with probability mass function $\PP[N = n] = p_n$. The following heuristic calculation explains why one would expect $\Delta_N / p_N$ to be unbiased:
\begin{align*}
    \EE\left[\frac{\Delta_N}{p_N}\right] 
   & = \EE\left[\EE\left[\frac{\Delta_N}{p_N} \bigg|~ N\right]\right]
    = \EE\left[\sum_{k=0}^\infty \frac{\Delta_k}{p_k} \cdot p_k\right] 
    = \sum_{k=0}^\infty \EE\left[\Delta_k\right] \\
    & = \lim_{k\rightarrow \infty} \EE\left[ g_{x_1}\left(\widebar{f(X_2)}_{n_k}\right)\right]
    =  g_{x_1}(\EE\left[f(X_2) \mid x_1\right])
\end{align*} 
where the third equality interchanges the order between  expectation and  (infinite) summation, the last equality interchanges the order between limit and expectation.

The above is the core idea of the unbiased MLMC estimator in \cite{rhee2015unbiased, Blanchet2015UnbiasedMC, Blanchet2019Unbiased}. However, it remains to justify several theoretical issues, such as the validity of the above interchange and the estimator's variance. An extra subtlety is the tradeoff between the sampling complexity and the variance. The expected sampling complexity for generating one estimator $\Delta_N/p_N$ is of the order of $\sum_{k=0}^\infty p_k n_k$. Clearly, it is desirable that the estimator has both finite variance and finite expected sampling complexity. 

Rhee and Glynn \cite{rhee2015unbiased} show the estimator $\Delta_N/p_N$ is unbiased and of finite variance if $\sum_{k=0}^\infty \EE\left[ \Delta_k^2\right]/p_k < \infty$ in a more general context. If one is  interested in estimating quantities of the form $g(\EE\left[X\right])$, Blanchet and Glynn
\cite{Blanchet2015UnbiasedMC} show one can choose $n_k = 2^k$ and $N\sim \text{Geo}(1 - 2^{-3/2})$ provided that $X$ has  bounded $6$-th order moment and $g$ is locally twice differentiable and grows moderately.  However, the assumption of \cite{Blanchet2015UnbiasedMC} is not satisfied even in this simple two-stage case, as the function $g_{x_1}(a) = \max\{f(x_1), a\}$ here is non-differentiable at $f(x_1)$.  The absence of smoothness assumptions on the function $g_{x_1}$ causes technical challenges  and calls for better theoretical guarantees in analyzing the unbiased MLMC estimator. 

Now we are ready to describe the two-stage MUSE and discuss its theoretical properties. Algorithm \ref{alg.one_step_MLMC} is referred to as the two-stage MUSE in contrast to the general/multi-stage MUSE described later. Roughly speaking, one first samples $x_1\sim \pi_1$, then constructs the standard unbiased MLMC estimator $\Delta_N/p_N$ for $g_{x_1}(\EE\left[f(X_2)\mid x_1\right])$ using a geometric random variable $N$ and $2^N$ $i.i.d.$ samples of $X_2$ with distribution $\pi_{2}(x_2\mid \{x_1\})$. The estimator $\Delta_n$ described in Step 4 is crucial for theoretical analysis. It is often referred to as the `antithetical difference' estimator in the literature \cite{Blanchet2015UnbiasedMC,Blanchet2019Unbiased}. The intuition is that the antithetic construction reduces the variance . As elaborated later in the theoretical analysis, the estimator $\Delta_N$ equals $0$ if both $\frac{S_{2^{N-1}}^E}{2^{N-1}}$ and $\frac{S_{2^{N-1}}^O}{2^{N-1}}$ are on the same side of $f(X_1(1))$. This observation turns out to be the key for controlling the expected computational complexity and variance simultaneously. We want to emphasize that the main contribution of the two-stage MUSE is more theoretical rather than the methodological. Algorithmically, the two-stage MUSE is very similar to the unbiased MLMC estimator. Theoretically, two-stage MUSE is the first unbiased estimator with theoretical guarantees for dealing with non-smooth functions.

\begin{algorithm}[h]
\SetAlgoLined
\textbf{Input}: A simulator $\mathcal S$ of a two-stage process $(X_1, X_2)$, parameter $r\in(1/2, 1)$. \\
\textbf{Output}: An unbiased estimator of $\EE \left[\max \{ f(X_1), \EE [f(X_2) \mid X_1]\}\right]$.\\
\textbf{Step 1}. Sample $N$ from geometric distribution \text{Geo}$(r)$.\\
\textbf{Step 2}. Sample $X_1(1)$. Conditioning on $X_1(1)$, sample $2^{N}$ $i.i.d.$  $X_2\left(1\right), \cdots,  X_2\left(2^N\right)$. \\
\textbf{Step 3}. Calculate the following three quantities:
\begin{align*}
        &S_{2^N} = f\left(X_2(1)\right) + \cdots + f\left(X_2\left(2^N\right)\right), &\\
        &S_{2^{N-1}}^O = f\left(X_2(1)\right) + f\left(X_2(3)\right) + \cdots + f\left(X_2\left(2^N-1\right)\right), & \text{(sum over odd indices)}\\
        &S_{2^{N-1}}^E = f\left(X_2(2)\right) + f\left(X_2(4)\right) + \cdots + f\left(X_2(2^N)\right). & \text{(sum over even indices)}
\end{align*}\\
\textbf{Step 4}. Calculate (note that $\Delta_0 := \max\left\lbrace f\left(X_1(1)\right), f\left(X_2(1)\right)\right\rbrace$)
\begin{align*}
        \Delta_N = &\max\left\lbrace f\left(X_1(1)\right), \frac{S_{2^N}}{2^N}\right\rbrace \\  &-\frac{1}{2}\left[\max\left\lbrace f\left(X_1(1)\right), \frac{S_{2^{N-1}}^O}{2^{N-1}}\right\rbrace + \max\left\lbrace f\left(X_1(1)\right), \frac{S_{2^{N-1}}^E}{2^{N-1}}\right\rbrace\right].
    \end{align*}\\
\textbf{Return}: $Y := \Delta_N / p_r(N)$.

\caption{Two-stage Multilevel Unbiased Stopping Estimator (Two-stage MUSE)}\label{alg.one_step_MLMC}
\end{algorithm}
Our main theoretical results on the two-stage MUSE are described in Theorem \ref{thm.one_step_unbiased}. Notice that the computational cost for Algorithm \ref{alg.one_step_MLMC} is a random variable depending on $N$. If we define the computation time for sampling one random variable and performing one arithmetic operation as `one unit', then the expected computational complexity is of the order of $\EE\left[2^N\right] = \sum\limits_{n=1}^\infty 2^n p_r(n)$.

\begin{theorem}\label{thm.one_step_unbiased}
Consider a two-stage process $(X_1, X_2)$. Suppose Assumptions  \ref{ass.simulation},  \ref{ass.moment} \emph{(}with $\delta < 1/4$\emph{)} and  \ref{ass.lipschitz_f} hold, and suppose Assumption \ref{ass.conditional} is satisfied with $T = 2$, i.e., 
\begin{equation}\label{aspn.bounded_density}
    \PP\left(\left| \EE[f(X_2) \mid X_1] - f(X_1) \right|\leq \varepsilon \right) < C\varepsilon
\end{equation}
 for all $\varepsilon >0$. Let $r = 1 - 2^{-\frac{2+9\delta/(80+40\delta)}{2+\delta/10}}\in (1/2, 1)$ in Algorithm \ref{alg.one_step_MLMC}. Then, the resulting estimator $Y$ in Algorithm \ref{alg.one_step_MLMC} has the following properties:
\begin{enumerate}[leftmargin = 0.6cm]
    \item [\emph{(1)}] $\EE [Y] = \EE \left[\max\left\lbrace f(X_1), \EE \left[f(X_2) \mid X_1\right]\right\rbrace\right] $.
    \item [\emph{(2)}] The expected computational complexity of $Y$ is finite.
    \item [\emph{(3)}] $\EE \left[|Y|^{2+\frac{\delta}{10}}\right]\leq \widetilde{C}\cdot L^{2+\delta}\left[1 + \EE\left[\|X_2\|^{2+\delta}\right]\right]$, where $\widetilde{C}$ is a constant independent of $(X_1, X_2)$.
\end{enumerate}
\end{theorem}

The proof of Theorem \ref{thm.one_step_unbiased} is deferred to the Appendix. As shown in Theorem \ref{thm.one_step_unbiased}, the two-stage MUSE is unbiased, has both finite $(2+\frac{\delta}{10})$-th moment (thus finite variance) and finite expected computational complexity. We also want to highlight a seemingly small theoretical improvement that turns out to be crucial in designing the multi-stage MUSE. In the existing literature, such as \cite{Blanchet2015UnbiasedMC,Blanchet2019Unbiased}, the estimator is guaranteed to have a finite second moment given the original random variable has a higher (say $6$-th) moment. In our case, we prove the estimator has $(2+\frac{\delta}{10})$-th moment given the original random variable has $(2+\delta)$-th moment, which makes the whole algorithm iterable in the multi-stage case. 

\subsection{MUSE for general optimal stopping problems}\label{sect.multi-stage MUSE}

In this section, we propose the multi-stage MUSE algorithm (Algorithm \ref{alg.multi_stage_MLMC}) which aims to provide an unbiased estimator for the general optimal stopping problem \eqref{def.optimal_stopping}. The multi-stage MUSE, as described in Algorithm \ref{alg.multi_stage_MLMC}, can be viewed as a recursive extension of the two-stage MUSE. To get an unbiased estimator of $U_T$, one  feeds 
$(0;\varnothing;\calS, r_1,\cdots, r_{T-1})$ into Algorithm \ref{alg.multi_stage_MLMC}. After sampling $x_1$ from the unconditioned distribution and $N_1\sim \text{Geo}(r_1)$, it suffices to construct $2^{N_1}$ unbiased estimators of $U_{T-1}(x_1)$ to build the MLMC estimator. Meanwhile, an unbiased estimator of $U_{T-1}(x_1)$ can be viewed as another optimal stopping problem with horizon $T-1$ and underlying process $\pi_{2:T}$, and therefore we call the same algorithm recursively after adding $x_1$ into the trajectory history. 

\begin{algorithm}[!h]
\SetAlgoLined
\textbf{Input}: Time index $k$. Trajectory history $H = \{x_1, \cdots, x_k\}$ or $\varnothing$. A simulator $\calS$ of the  conditional distribution $\pi_{T-k}$, parameters $r_{k+1}, \cdots, r_{T-1} \in (1/2, 1)$. \\
\textbf{Output}: An unbiased estimator of $U_{T-k}$ in \eqref{eqn.utility}.\\

\If{ $k = T-1$ } {Sample one $x_T$ from the conditional distribution of $\pi_T$ given $H$.\\ 
\textbf{Return $Y := f(x_T)$}. }
\Else{
Sample $x_{k+1}$ from the condition distribution $\pi_{k+1}$ given $H$.\\
Add $x_{k+1}$ to the trajectory history $H$.\\
Sample $N_{k+1}\sim \text{Geo}(r_{k+1})$.\\
Call Algorithm \ref{alg.multi_stage_MLMC} for $2^{N_{k+1}}$ times with inputs $(H; \calS, r_{k+2} \cdots, r_{T-1})$, label the outputs by $Y_{k+1}(1), \cdots, Y_{k+1}(2^{N_{k+1}})$.\\
Calculate the following three quantities:
\begin{align*}
        &S_{2^{N_{k+1}}} = Y_{k+1}(1) + \cdots + Y_{k+1}(2^{N_{k+1}}),\\
        &S_{2^{N_{k+1} - 1}}^O = Y_{k+1}(1) + Y_{k+1}(3) \cdots + Y_{k+1}(2^{N_{k+1}} - 1),\\
        &S_{2^{N_{k+1} - 1}}^E = Y_{k+1}(2) + Y_{k+1}(4) \cdots + Y_{k+1}(2^{N_{k+1}} ). 
\end{align*}\\
 Calculate (note that $\Delta_0 := \max\left\lbrace f(x_{k+1}), Y_{k+2}(1)\right\rbrace$)
\begin{align*} 
        \Delta_{N_{k+1}} = &\max\left\lbrace f(x_{k+1}), \frac{S_{2^{N_{k+1}}}}{2^{N_{k+1}}}\right\rbrace \\
        &- \frac{1}{2}\left[\max\left\lbrace f(x_{k+1}), \frac{S_{2^{N_{k+1}-1}}^O}{2^{N_{k+1}-1} }\right\rbrace + \max\left\lbrace f(x_{k+1}), \frac{S_{2^{N_{k+1}-1}}^E}{2^{N_{k+1}-1} }\right\rbrace\right].
    \end{align*}\\
\textbf{Return}: $Y := \Delta_{N_{k+1}} / p_{r_{k+1}}\left(N_{k+1}\right)$. 
}
\caption{Multi-stage Multilevel Unbiased Stopping Estimator (Multi-stage MUSE)}\label{alg.multi_stage_MLMC}
\end{algorithm}

The next theorem studies the theoretical properties of the multi-stage MUSE. The computational complexity of Algorithm \ref{alg.multi_stage_MLMC} comes from the sampling complexity, which is of the order of $\EE\left[\prod_{k=1}^{T-1}2^{N_k}\right]$.

\begin{theorem}\label{thm.multi_stage_unbiased}
With Assumptions \ref{ass.simulation}, \ref{ass.moment}, \ref{ass.lipschitz_f}, and  \ref{ass.conditional}, consider the input $$(0;\varnothing;\calS, r_1,\cdots, r_{T-1})$$ in Algorithm \ref{alg.multi_stage_MLMC}, where 
\begin{equation*}
r_i = 1 - 2^{-\frac{2+9\left(\delta\cdot10^{i+1-T}\right)/\left(80+40\left(\delta\cdot10^{i+1-T}\right)\right)}{2+ \delta\cdot10^{i-T}}}\in (1/2, 1)
\end{equation*}
for $1\leq i \leq T-1$. Then, the resulting estimator $Y$ in Algorithm \ref{alg.multi_stage_MLMC} has the following properties:
\begin{enumerate}[leftmargin = 0.6cm]
    \item [\emph{(1)}] $\EE [Y] = U_T$.
    \item [\emph{(2)}] The expected computational complexity is $O\left(10^{T^2}\right)$. 
    \item [\emph{(3)}] $\var\left(Y\right) = O\left(10^{T^2}\right)$. 
\end{enumerate}
\end{theorem}

To illustrate the iterative structure of multi-stage MUSE, we sketch the proof of Theorem \ref{thm.multi_stage_unbiased} in below. The detailed proof is deferred to the Appendix. By the standard dynamical programming for optimal stopping, we have
\begin{equation*}
\begin{cases}
U_1(X_{1:T-1})= \EE \left[f(X_T) \mid X_{1:T-1}\right],  \\
U_{T-k}(X_{1:k}) = \EE \left[\max \left\lbrace f\left(X_{k+1}\right), U_{T-(k+1)}(X_{1:k+1})\right\rbrace \mid X_{1:k}\right], \quad 0\leq k\leq T-2.
\end{cases}
\end{equation*}
Here, for a generic $d$-tuple $ (v_1,\cdots, v_d)$, let $v_{i:j}:=(v_i,\cdots, v_j)$ for $1\leq i\leq j\leq d$ for notational convenience. By applying the techniques in the proof of Theorem \ref{thm.one_step_unbiased}, one can show that for each stage, the output $Y_{T-k}$ always has a moment of order greater than $2$, and is  unbiased for $U_{T-k}$. The proof of the moment bounds is iterable because of the careful technical analysis in Theorem \ref{thm.one_step_unbiased}. Moreover, by a proper choice of the parameters $r_1,\cdots, r_{T-1}$, the expected sampling complexity for each stage is bounded. As a result, the total expected sampling complexity is also bounded.

\begin{corollary}\label{cor.application_of_unbiasedness}
Let Assumption \ref{ass.moment} and Assumption \ref{ass.conditional} hold. For any $\varepsilon >0$, and a fixed time horizon $T$, we can construct an estimator $Y$ that satisfies the following properties:
\begin{itemize}
    \item The expected computational complexity for constructing $Y$ is $O(1/\varepsilon^2)$.
    \item The expected mean squared error between $Y$ and the true utility is bounded by $\epsilon^2$, i.e., $\EE \left[\left( Y - U_T\right)^2\right]\leq\varepsilon^2.$
\end{itemize}
 
\end{corollary}
\begin{proof}[Proof of Corollary \ref{cor.application_of_unbiasedness}]
We fix a positive integer $n$. Calling Algorithm \ref{alg.multi_stage_MLMC} $n$ times yields $n$ $i.i.d.$ unbiased estimators $Y_1,\cdots, Y_n$ of $U_T$. Then, 
$$
    \EE \left[\left(\frac{1}{n}\sum_{i=1}^nY_i - U_T\right)^2\right] = \EE \left[\left(\frac{1}{n}\sum_{i=1}^n (Y_i - \EE[Y_i])\right)^2\right] = \frac{1}{n}\var(Y_1).
$$
Taking $n = \var(Y_1)/\varepsilon^2$ (note that $\var(Y_1) < \infty$ by Theorem \ref{thm.multi_stage_unbiased}) and define $Y:= \frac{1}{n}\sum_{i=1}^nY_i $. It follows from the above calculation that  $\EE \left[\left( Y - U_T\right)^2\right]\leq\varepsilon^2.$ Moreover, since sampling each $Y_i$ has expected computational complexity $O(1)$,  the expected computational complexity for constructing $Y$ is $O(1/\varepsilon^2)$, as desired.
\end{proof}

Finally we comment on some practical issues when implementing the MUSE for multi-stage optimal stopping problems. One drawback of our algorithm is that the computational complexity (the constant hidden in $\calO(1/\epsilon^2)$ in Corollary \ref{cor.application_of_unbiasedness}) grows exponentially with time horizon $T$. Therefore, our algorithm is prohibitively slow when $T$ becomes large. We believe this is expected due to the comprehensive multi-stage structure of the optimal stopping problem \eqref{def.optimal_stopping}.  In fact, the same phenomenon happens in the Monte Carlo-based methods, including the popular algorithms of  Broadie and Glasserman \cite{broadie1997pricing}, Longstaff and Schwartz \cite{longstaff2001valuing} and Tsitsiklis and Van Roy \cite{tsitsiklis2001regression}. It is known in Glasserman and Yu \cite{glasserman2004number} that the number of sample paths required for the regression coefficients to converge grows exponentially in the degree of basis functions under the worst-case scenario.  Zanger \cite{zanger2013quantitative} proved the expected $L^2$ error has an $O((\log^{1/2}N) N^{-1/2})$ convergence rate ($N$ is the number of sample paths) given the approximation architecture has a finite Vapnik–Chervonenkis (VC) dimension. Their error bound also scales exponentially with respect to the time horizon, see Theorem 3.3 of \cite{zanger2013quantitative}.  Meanwhile, we emphasize that besides Assumption \ref{ass.conditional},  there is currently no specific distributional assumption on the underlying process. Therefore, there is a potential for designing computationally efficient  estimators given additional distribution assumptions. Moreover, if the unbiased requirement can be relaxed, then it is possible to design fast algorithms while retaining the $\calO(1/\epsilon^2)$ complexity. One can potentially either choose a fixed level based on the standard MLMC approaches \cite{giles2015multilevel}, or use a truncated geometric random variable to replace the geometric random variable $N$ in Algorithm \ref{alg.multi_stage_MLMC} as described in the recent work \cite{asi2021stochastic}. These ideas open up exciting possibilities for new algorithms, but are already beyond the target of our paper.
 
\subsection{Confidence Interval and Optimal Stopping Time}\label{sect.applications}
A confidence interval (CI) is crucial if one is not merely interested in getting a point estimate, but also expects to assess the quality of such estimation. Fortunately, since many $i.i.d.$ estimators of $U_T$ can be constructed by repeatedly calling Algorithm \ref{alg.multi_stage_MLMC}, the   $1-\alpha$ confidence interval (CI) of the utility can be constructed as follows: Let $Y_1, \cdots, Y_n$ be $n$ unbiased estimators of $U_T$ generated by the MUSE. Let $\widebar{Y}$ be their empirical mean and $s$ the standard deviation. Then, two types of CIs  can be built via
 \begin{itemize}[leftmargin = *]
     \item (CLT) $\left[\widebar{Y} - z_{\alpha/2}\cdot s/\sqrt{n}, \widebar{Y} + z_{\alpha/2}\cdot s/\sqrt{n}\right]$, where $z_{\alpha/2}$ is the $(1-\alpha/2)$-th quantile of  $\calN(0,1)$.
     \item (Bootstrap \cite{efron1994introduction}) $[\widebar Y^\star_{\alpha/2}, \widebar Y^\star_{1-\alpha/2}]$, where $\widebar Y^\star_{\alpha/2}, \widebar Y^\star_{1-\alpha/2}$,  are the $\alpha/2$-th and $(1-\alpha/2)$-th  empirical quantile of the bootstrape averages.
 \end{itemize}
 
In principle, both methods are valid as the number of simulated estimators goes to infinity. The first CI is based on the Central Limit Theorem, and the convergence rate depends on the higher-order cumulants. The second CI uses the empirical distribution to approximate the actual underlying distribution, which is non-parametric and is (monotone) transformation-respecting (\cite{efron1994introduction}, Chapter 12). It is known that the percentile bootstrap may not work well when the data has a significant skewed distribution. In these cases one may consider alternative methods such as the $\text{BC}_\text{a}$ (bias corrected accelerated) bootstrap \cite{diciccio1996bootstrap}.
 
Besides estimating $U_T$, we are also interested in finding the optimal stopping time $\tau^*$ such that $\EE\left[f\left(X_{\tau^*}\right)\right] = U_T$. By standard dynamical programming, 
    $$\tau^* = \inf \left\lbrace k\geq 1: U_{T-k}(X_1,\cdots, X_k) \leq  f(X_k)\right\rbrace.$$
Though $U_{T-k}$ is not analytically available, the MUSE provides us with powerful tools for estimating $U_{T-k}$ at each round. The algorithm for the optimal stopping time is as follows:

\begin{algorithm}[H]
\SetAlgoLined
\textbf{Input}: Simulator of the process $(X_1,\cdots, X_T)$, tolerance level $\varepsilon$.\\
\textbf{Output}: A stopping time $\widehat{\tau}$.\\
Sample $X_1 = x_1$.\\
\For{$k\gets 1$ \KwTo $T-1$}{
Call Algorithm \ref{alg.multi_stage_MLMC} with history $(x_1,\cdots, x_k)$  $n$ times to get $i.i.d.$ unbiased estimators $Y_1,\cdots, Y_n$ of $U_{T-k}$. \\
\textbf{if} $f(x_k) > \widebar{Y} - \varepsilon$ \textbf{ return} $k$. \textbf{else} Sample $X_{k+1} = x_{k+1}$.\\
\textbf{if} $k+1 = T$  \textbf{ return} $T$.
}
\caption{Optimal Stopping Time via MUSE}\label{alg.optimal_stopping_time}
\end{algorithm}
There are multiple ways of choosing $\varepsilon$, which clearly depend on the decision maker's risk sensitivity. One promising option would be to choose  $\varepsilon$ adaptively, according to the CIs derived by the MUSE. 

\section{Numerical Experiments} \label{sect.numerical_experiments}
\subsection{Optimal Stopping of Independent Random Variables}\label{sect.i.i.d.}
The optimal stopping problem for independent random variables has been extensively studied in the literature.
In this example, we consider the case where $X_1,\cdots, X_T$ are $i.i.d$. $\calN (0,1)$ random variables with reward  $f(x) = x$. Standard calculation yields  $U_1 = 0$, $U_2 = \EE\left[\lvert X_1\rvert\right]/2$ and $U_k = \EE\left[\max\{{X_1}, U_{k-1}\}\right]$ so that the utility can be solved numerically. With each fixed time horizon, three estimators -- MUSE and two vanilla Monte Carlo estimators MC1 and MC2 are implemented.  MC1 is a naive Monte Carlo estimator. For each $T$, it samples $10^7$ paths and estimates $U_T$ by the average of the maximum in each path, which is clearly biased. MC2 is a refinement of MC1 but still  biased. It samples tree-like paths as described in Section \ref{sect.intro}, Figure \ref{fig.Tree}. In our case, the simulated data forms a forest that consists of $1000$ complete $5$-ary trees of depth $T$. MC2 estimates the utility using the dynamical programming formula $U_{T} = \EE \left[\max \left\lbrace X_{1}, U_{T-1}(X_1)\right\rbrace \right]$ in a backward recursive way. Given the history $X_1, \cdots, X_{n-1}$, the utility of $U_{1}(X_1, \cdots, X_{n-1})$ can be easily estimated by averaging the samples in the last layer. Similarly, we can use the formula $U_{2}(X_1,\cdots, X_{n-2}) =  \EE \left[\max \left\lbrace X_{n-1}, U_{1}(X_1,\cdots, X_{n-1})\right\rbrace \right]$ to estimate $U_2$ after replacing the quantity $U_{1}(X_1,\cdots, X_{n-1})$ by its estimator described above. Then we estimate $U_3, U_4, \cdots$ and finally $U_T$. Formally, the final estimator $\hat U_T = (\sum_{i=1}^{1000} \hat U_T^{(i)} )/1000$ is the average of the $1000$ estimators from each tree. For each $i$, the estimator $\hat U_T^{(i)}$ of tree $i$ is of  the form $\hat U_{T}^{(i)}:= (\sum_{j = 1}^5 \max\{X^{(i)}_{1,1}, \hat U_{T-1}^{(i)}\})/5$, where the number $5$ comes from the $5$-ary tree design,  $X^{(i)}_{1,1}$ is the root node of the $i$-th tree, and $\hat U_{T-1}^{(i)}$ is the estimator for $U_{T-1}(X^{(i)}_{1,1})$ using the dynamical programming procedure mentioned above. 

The only hyperparameter for the MUSE is the success probability $r$ for the geometric distribution. Larger $r$ leads to shorter computational time but larger variance, and vice versa. We implement a simple experiment to determine $r$. For each $r$ in $\{0.51, 0.52, \ldots, 0.7\}$, we run $10^6$ MUSEs for horizon $T = 3$ and examine their empirical performances. Our results are summarized in Figure \ref{fig.SuccessProb}. The cost of time decays significantly when $r$ increases. Furthermore, we also calculate the self-normalized variance \cite{Blanchet2015UnbiasedMC} as a measure of efficiency.  The self-normalized  variance is defined as the product between the expected time and the variance for every single estimator. It is clear from the right subplot of Figure \ref{fig.SuccessProb} that the self-normalized variance initially decays and then increases as $r$ increases, with a minimum at around $0.6$, therefore we choose $r = 0.6$ in the numerical experiments henceforth.

\begin{figure}[!htbp]
    \centering
    \includegraphics[height = 7cm, width=0.9\textwidth]{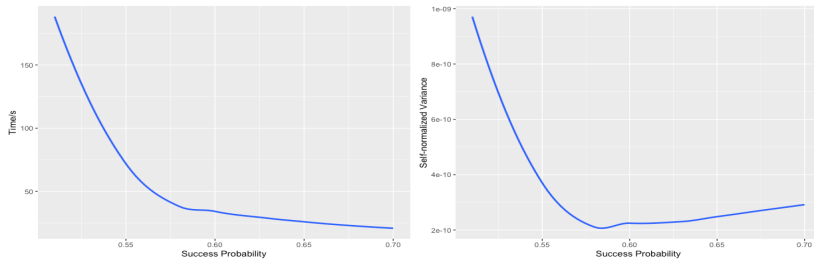}
    \caption{Left: The cost of time for generating $10^6$ MUSEs with different success probabilities. 
    Right: The self-normalized variance of the MUSEs with different success probabilities.}
    \label{fig.SuccessProb}
\end{figure}

After setting up the hyperparameter, we implement the three methods for $T\in\{2,\cdots, 7\}$. Our results are presented in Figure \ref{fig.MUSE_iid}. Both MC1 (red curve) and MC2 (green curve) systematically overestimate the true utility (black dotted line), as expected. The accuracy of MC1 is poor while MC2 has much better accuracy, sometimes comparable with the MUSE. The MUSE (blue curve) uses parameters $r_i = 0.6$ for each stage\footnote{To ease the computation burden, the parameters chosen here do not strictly follow Theorem \ref{thm.multi_stage_unbiased}. }, and averages of $10^6$ estimators for each $T$. It typically has  the most accurate result among all three methods.  To better understand the empirical convergence behavior of the MUSE, we also show the traceplot for the running average of the MUSE for each horizon in the right subplot of Figure \ref{fig.MUSE_iid}. It is clear from the traceplot that the CIs typically covers the ground truth, though
the convergence becomes much slower when $T$ is increasing. 

\begin{figure}[!htbp]
    \centering
    \includegraphics[height = 6cm, width=\textwidth]{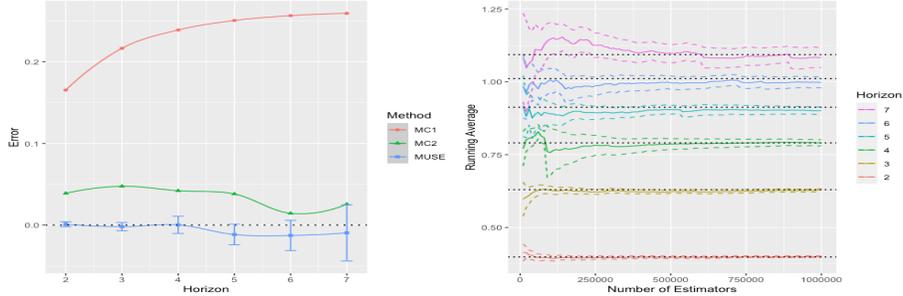}
    \caption{Left: Comparison between the errors of the MUSE (blue), MC1(red), and MC2(green) for estimating the utility for $i.i.d.$ standard Gaussian random variables. Blue error bars stands for the $95\%$ confidence intervals of the MUSE.  Black dotted line stands for the ground truth (error = $0$).
    Right: The traceplot of the running averages of the MUSE with different horizons. Black dotted line stands for the ground truth. Colored dashed lines stands for the running $95\%$ CIs of the MUSE}
    \label{fig.MUSE_iid}
\end{figure}

\subsection{Pricing the Bermudan options with high-dimensional inputs on a computer cluster}\label{sect.Bermudan-basket}
In this section we consider a more challenging setup, where the underlying process $\mathbf{X}_t := \left(X_t^{(1)}, \cdots, X_t^{(d)}\right)$ takes values in a high-dimensional space $\mathbb R^d$. The example we are considering here is a standard one -- pricing the high-dimensional Bermudan-basket put options. The underlying process is a $d$-dimensional independent geometric Brownian motion with drift $\gamma-\delta$ and volatility $\sigma$ where all parameters will be specified later. Bermudan-basket option has utility $f(t, \bX_t) = e^{-\gamma t} \max\{0, K - \sum_{i=1}^d X_t^{(i)}/d\}$ at each $t$, where $K$ is the strike price and $e^{-\gamma}$ is often referred to as the discounting factor. Bermudan option is only exercisable in a discrete set of times, which transforms the pricing problem to solving the optimal stopping problem:
$U_T :=\sup_{\tau\in\{T_1,\cdots T_k\}} \EE \left[f\left(\tau, \mathbf{X}_\tau\right)\right]$, where $0\leq T_1 \leq \cdots \leq T_k\leq T$ are all the  exercisable dates. It has been observed \cite{herrera2021optimal}  that the computational cost for standard regression-based methods typically scales superlinearly with dimension $d$, which discourages their uses in the high-dimensional setups. Existing experiments on Bermudan options often 
assumes $d\leq 20$, though it can be as large as $5000$ in practice  \cite{becker2019solving}. 

In our experiment we adopt the standard parameters in \cite{bender2006policy, jain2012pricing} where $T = 3$ (years), $\sigma = 0.2, \gamma = 0.05, \delta = 0, K = X^{(i)}_0 = 100$ for every $i$. Owners can exercise the option at the initial time or after $1,2,3$ years. We first benchmark our result with the results reported in \cite{bender2006policy, jain2012pricing} when  $d = 5$, next we present our results for $d \in \{10,20,100,1000\}$. For each $d$, we use $10^7$ MUSEs generated by a $500$-core CPU-based computer cluster, where the parameters $r_i$ are set to be $0.6$ for each stage. The results  when $d = 5$ is presented in Table \ref{tab.d=5}, the MUSE matches the results from other methods while preserving unbiasedness and having a relatively small standard error.

\begin{table}[htbp!]
\scalebox{0.9}{
\begin{tabular}{|l|l|l|l|l|l|}
\hline
\multicolumn{1}{|c|}{Method} & \multicolumn{1}{c|}{\begin{tabular}[c]{@{}c@{}}LSM \\ (s.e.)\end{tabular}} & \multicolumn{1}{c|}{\begin{tabular}[c]{@{}c@{}}SGM direct\\     (s.e.)\end{tabular}} & \multicolumn{1}{c|}{\begin{tabular}[c]{@{}c@{}}SGM LB\\   (s.e.)\end{tabular}} & \multicolumn{1}{c|}{\begin{tabular}[c]{@{}c@{}}BKS\\ ($95\%$ CI)\end{tabular}} & \multicolumn{1}{c|}{\textbf{\begin{tabular}[c]{@{}c@{}}MUSE\\  (s.e.)\end{tabular}}} \\ \hline
$d = 5$                      & $2.163 (0.001)$                                                            & $2.141 (0.008)$                                                                      & $2.134 (0.012)$                                                                & $[2.154, 2.164]$                                                               & {$\mathbf{2.161 (0.004)}$}                                                             \\ \hline
\end{tabular}}
\caption{Comparison between different methods when $d = 5$. SGM and BKS stands for results reported by \cite{jain2012pricing} and \cite{bender2006policy} respectively. LSM stands for Longstaff–Schwartz method, reported by \cite{jain2012pricing}.}\label{tab.d=5}
\end{table}

Table \ref{tab.different d} records the estimates and the standard errors of the MUSE when $d$ is increasing. There are no existing benchmark results for large $d$ thus we are not able to compare with the ground truth. But the law of large numbers shows the utility should converge to $0$ as $d$ goes to infinity, which matches our result here. We also record the average computing time for every processor in the last column of Table \ref{tab.different d}, the computation time scales sublinearly with the dimensionality $d$, which may be benefited from the use of vectorization in simulating the $d$-dimensional geometric Brownian motion. We also plot the histogram of the computing time among $500$ 
cores when $d = 100$ in Figure \ref{fig.hist d=100}. It is clear from Figure \ref{fig.hist d=100} that the computing times are relatively short (less than $15$ seconds) for most clusters even in this high-dimensional regime. There are a small proportion of clusters that uses much longer time. This fact indicates the MUSE has a high variance in its computational complexity, which is  in line with our theoretical intuitions. Finally, it seems the MUSE scales well with $d$, which may be another appealing feature besides parallel computing.

\begin{figure}[htbp!]
\begin{floatrow}

\capbtabbox{%
\scalebox{0.8}{%
 \begin{tabular}{|c|c|c|}
\hline
$d$  & \begin{tabular}[c]{@{}c@{}}MUSE \\ (s.e.)\end{tabular} & \begin{tabular}[c]{@{}c@{}}Average Time (s) \\ per processor\end{tabular} \\ \hline
$5$    & \begin{tabular}[c]{@{}c@{}}$2.161$ \\ ($0.004$)\end{tabular}                                        & $15.922$                                                                  \\ \hline
$10$   & \begin{tabular}[c]{@{}c@{}}$0.985$ \\ ($0.002$)\end{tabular}                                      & $14.787$                                                                  \\ \hline
$20$   & \begin{tabular}[c]{@{}c@{}}$0.355$ \\ ($0.001$)\end{tabular}                                     & $16.004$                                                                  \\ \hline
    $100$  & \begin{tabular}[c]{@{}c@{}}$0.0043$ \\ ($<10^{-4}$)\end{tabular}                        & $18.271$                                                                  \\ \hline
$1000$ & $0 (0)$                                                & $32.191$                                                                  \\ \hline
\end{tabular}
}{%
  \caption{Results of the MUSE under different dimensions. The second column reports the means and standard errors of the MUSEs. The third column reports the average computation time over the $500$ processors. }\label{tab.different d}%
}}

\ffigbox{%
  \includegraphics[width = 5.5cm]{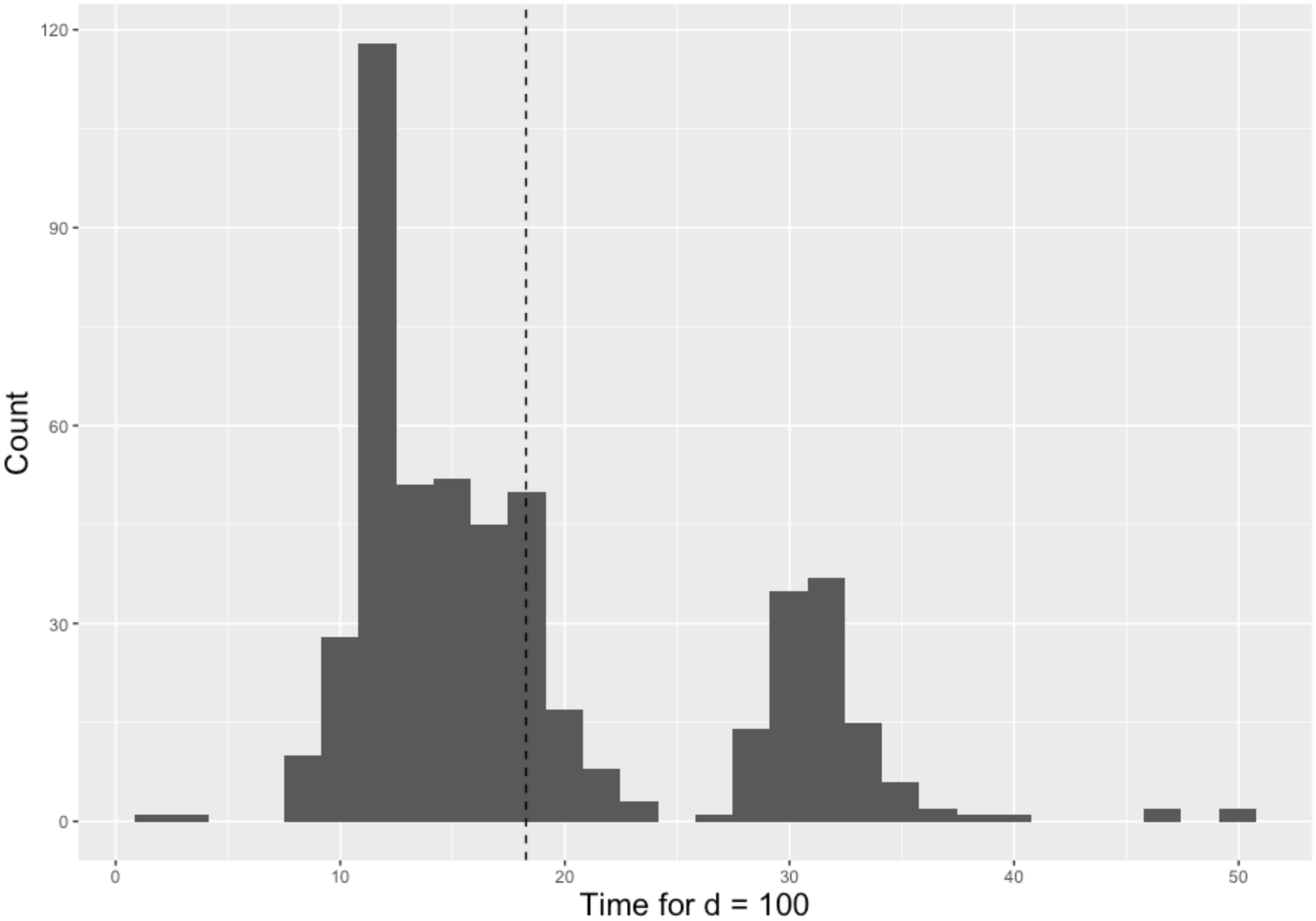}%
}{%
  \caption{Histogram of computational times among $500$ processors when $d = 100$. The black dotted line is the average computation time, which is $18.271$ seconds.}\label{fig.hist d=100}%
}
\end{floatrow}
\end{figure}

\section{Conclusion and Future Work} \label{sect.conclusion}
Optimal stopping problems play an important role in modern decision-making processes. However, existing simulation algorithms introduce unavoidable bias in estimating the utility. In this paper, an unbiased estimator, the MUSE, is proposed and analyzed. Our estimator is easy to implement and enjoys unbiasedness, finite variance, and finite computational complexity after choosing the parameters appropriately. A key ingredient of the general MUSE is the iterative use of the two-stage MUSE, which preserves unbiasedness at every stage by the multilevel approach.

In the theoretical part of this paper, we focus on bounding the variance and computational complexity of the MUSE. Though finite variance and finite complexity are guaranteed, these upper bounds may be too crude to shed light on practical applications. Moreover, theoretical guarantees on the applications described in Section \ref{sect.applications}, such as regret bounds for Algorithm \ref{alg.optimal_stopping_time}, are worth investigating.

In the numerical studies,  experiments in Section \ref{sect.numerical_experiments} suggest the MUSE is able to provide accurate estimation for the utilities, especially when $T$ is small or moderate. The MUSE also seems to scale well with the dimensionality of the underlying process, as shown in Section \ref{sect.Bermudan-basket}.  On the other hand, our estimator's variance and computational complexity grow significantly with the horizon length. Running Algorithm \ref{alg.multi_stage_MLMC} can easily be prohibitive with large horizons. It remains a key challenge to design scalable algorithms while maintaining unbiasedness, or at least controlling bias at a negligible level under the large-horizon regime. As a closing remark, unbiasedness is undoubtedly an appealing property in parallel computation, but it could come with higher computational cost or lower statistical accuracy. Therefore, studying the trade-offs between unbiasedness, computational budget constraints, and accuracy may be of paramount interest to both theorists and practitioners. We hope future studies will provide much-needed insight toward achieving practical unbiasedness with sustainable cost and high accuracy.

While our paper focus on the optimal stopping problem, we believe our technique can potentially be extended to more general setups. The optimal stopping is a subfamily of the stochastic control problems, where one picks the optimal time to stop. One natural extension is the case where the one needs to choose one out of $K$ possible actions at each stage. Moreover, our paper considers the optimal problem with finite horizon. Another natural extension is to consider the infinite horizon stopping problem with discounted reward.

\section{Acknowledgement}
Material in this paper is based upon work supported by the Air Force Office of Scientific Research under award number FA9550-20-1-0397. Additional support is gratefully acknowledged from NSF grants 1915967, 1820942, 1838576, 2210849. Guanyang Wang would like to sincerely thank Changpeng Lu for her patient guidance on configuring the computer cluster, and David Sichen Wu for helpful suggestions on improving this paper. We would like to thank  two anonymous referees for their time, effort, and  constructive suggestions.

\appendix
\input{supplementary}

\bibliography{di}

\clearpage

\end{document}

%% file: supplementary.tex
\section{Auxiliary Results}
\begin{lemma}[\cite{marcinkiewicz1937quelques} Marcinkiewicz-Zygmund inequality]\label{lemma.M-Z_inequality}
If $X_1,\cdots, X_n$ are independent random variables with $\EE[X_i] = 0$ and $\EE\left[|X_i|^p\right]<\infty$ for some $p>2$. Then,
\begin{equation*}
    \EE \left[\left|\sum_{i=1}^nX_i\right|^p\right]\leq C_p'\EE \left[\left(\sum_{i=1}^n|X_i|^2\right)^{p/2}\right],
\end{equation*}
where $C_p'$ is a constant that only depends on $p$. If we further assume that $X_1, \cdots, X_n$ are i.i.d.. Then, 
\begin{equation*}
    \EE \left[\left|\frac{1}{n}\sum_{i=1}^nX_i\right|^p\right]\leq C_p' \EE \left[\frac{1}{n^{p/2}}\left[\frac{1}{n}\sum_{i=1}^n|X_i|^2\right]^{p/2}\right]\leq  C_p'\cdot\frac{\EE|X_1|^p}{n^{p/2}}.
\end{equation*}
\end{lemma}

\begin{corollary}\label{cor.conditional_M-Z_inequality}
Let $(Z_1, Z_2)$ be a $2-$stage stochastic process, and there exists $p>2$, such that $\sup_{i=1,2}\EE\left[|Z_i|^p\right]<\infty$. Conditioning on $Z_1$, sample i.i.d. $Z_{2}(1),\cdots, Z_{2}(n)$. Then, 
\begin{equation*}
    \EE \left[\left|\frac{1}{n}\sum_{i=1}^n Z_2(i) - \EE[Z_2 \mid Z_1]\right|^p\right] \leq C_p \cdot\frac{\EE\left[|Z_2|^p\right]}{n^{p/2}},
\end{equation*}
where $C_p$ is an universal constant only depends on $p$.
\end{corollary}

\begin{proof}[Proof of Corollary \ref{cor.conditional_M-Z_inequality}]
Since $\EE\left[|Z_2|^p\right] < \infty$, we have that $\EE\left[|Z_2|^p \mid Z_1 = z_1\right]$ exists almost surely. Let $\pi_{1:2}$ be the joint measure of $(Z_1, Z_2)$, applying Lemma \ref{lemma.M-Z_inequality} to the conditional distribution $\pi_2(\cdot \mid z_1)$ yields
\begin{align*}
    \EE \left[\left|\frac{1}{n}\sum_{i=1}^nZ_2(i) - \EE[Z_2 \mid Z_1]\right|^p\right] &= \int_{\Omega}\EE \left[\left|\frac{1}{n}\sum_{i=1}^n Z_2(i) - \EE [Z_2 \mid Z_1]\right|^p \;\middle|\; Z_1 = z_1\right]\pi_1(dz_1)\\
    &\leq \int_{\Omega} \frac{C_p'}{n^{p/2}} \EE \left[|Z_2 - \EE[Z_2\mid Z_1]|^p \mid Z_1 = z_1\right]\pi_1(dz_1)\\
    &\leq \int_{\Omega} \frac{C_p'2^{p-1}}{n^{p/2}} \left[\EE \left[|Z_2|^p | Z_1 = z_1 \right] + \left|\EE[Z_2 | Z_1 = z_1]\right|^p\right]\pi_1(dz_1)\\
    &\leq \int_{\Omega} \frac{C_p'2^p}{n^{p/2}} \EE \left[|Z_2|^p \mid Z_1 = z_1 \right] \pi_1(dz_1) = \frac{C_p'2^p}{n^{p/2}}\EE[|Z_2|^p].
\end{align*}
\end{proof}

\section{Proofs of Main Theorems}
 We first present the proof of Theorem \ref{thm.one_step_unbiased}. 
\begin{proof}[Proof of Theorem \ref{thm.one_step_unbiased}]
We first show that  $\EE[Y] = \EE \left[\max\left\lbrace f(X_1), \EE \left[f(X_2) \mid X_1\right]\right\rbrace\right]$. Note that the $X_1, X_2$ are integrable, 
\begin{align}
    \nonumber&\EE [Y] \\
    = \quad & \EE \left[\EE \left[Y \mid N\right]\right] \label{eq.proof_thm1_exchange_1}\\
    = \quad & \sum_{n=0}^{\infty} \EE [\Delta_n] \label{eq.proof_thm1_exchange_2}\\
    \nonumber= \quad  & \sum_{n=1}^{\infty} \left(\EE\left[\max\left\lbrace f\left(X_1(1)\right), \frac{S_{2^n}}{2^n}\right\rbrace\right] - \EE\left[\max\left\lbrace f\left(X_1(1)\right), \frac{S_{2^{n-1}}}{2^{n-1}}\right\rbrace\right] \right) + \\
    \nonumber&\qquad \EE \left[\max\left\lbrace f\left(X_1(1)\right), f\left(X_2(1)\right)\right\rbrace\right]\\
    \nonumber= \quad & \lim_{n\rightarrow \infty} \EE \left[\max\left\lbrace f\left(X_1(1)\right), \frac{S_{2^n}}{2^n}\right\rbrace\right] - \EE \left[f\left(X_1(1)\right), S_1\right] + \EE \left[\max\left\lbrace f\left(X_1(1)\right), f\left(X_2(1)\right)\right\rbrace\right]\\
    = \quad & \EE \left[\max\left\lbrace f\left(X_1(1)\right), \lim_{n\rightarrow \infty}\frac{S_{2^n}}{2^n}\right\rbrace\right]\label{eq.proof_thm1_lln_1}\\
    = \quad & \EE \left[\max\left\lbrace f\left(X_1(1)\right), \EE[f(X_2) \mid X_1(1)]\right\rbrace\right]. \label{eq.proof_thm1_lln_2}
\end{align}
Here the law of large number is applied to guarantee the equality between \eqref{eq.proof_thm1_lln_1} and \eqref{eq.proof_thm1_lln_2}, and the equality of \eqref{eq.proof_thm1_exchange_1}, and \eqref{eq.proof_thm1_exchange_2} is established by the interchanging the order of summation and expectation, which is legitimate due to the fact that 
\begin{equation}\label{ieq.condition_interchange_sum_expectation}
    \sum_{n=0}^n \EE \left[|\Delta_n|\right] < \infty.
\end{equation}
To verify the inequality \eqref{ieq.condition_interchange_sum_expectation}, note that $\max\{x, a\}$ is a 1-Lipschitz function of $x$ for any fixed $a$, we have 
\begin{align*}
    |\Delta_n|&\leq \frac{1}{2}\left|\max\left\lbrace f\left(X_1(1)\right), \frac{S_{2^n}}{2^n}\right\rbrace - \max\left\lbrace f\left(X_1(1)\right), \frac{S_{2^{n-1}}^O}{2^{n-1}}\right\rbrace\right| \\
    &\qquad + \frac{1}{2}\left|\max\left\lbrace f\left(X_1(1)\right), \frac{S_{2^n}}{2^n}\right\rbrace - \max\left\lbrace f\left(X_1(1)\right), \frac{S_{2^{n-1}}^E}{2^{n-1}}\right\rbrace\right|\\
    &\leq \frac{1}{2}\left|S_{2^{n-1}}^O/2^{n-1} - S_{2^{n-1}}^E/2^{n-1}\right|.
\end{align*}
By Corollary \ref{cor.conditional_M-Z_inequality}, and note that $|f(x)|\leq L(1 + \|x\|)$, we have 
\begin{align*}
    \sum_{n=0}^\infty \EE \left[|\Delta_n|\right]
    & \leq 2L^{2+\delta}\left[1 + \sup_{i=1,2} \left[\EE\left[\|X_i\|^{2+\delta}\right]\right]^{\frac{1}{2+\delta}}\right] +  \frac{1}{2}\sum_{n=1}^\infty  \left[\EE\left[\left|\frac{S^O_{2^{n-1}}}{2^{n-1}} - \frac{S^E_{2^{n-1}}}{2^{n-1}}\right|^{2+\delta}\right]\right]^{\frac{1}{2+\delta}}\\
    &\leq 2L^{2+\delta}\left[1 + \sup_{i=1,2} \left[\EE\left[\|X_i\|^{2+\delta}\right]\right]^{\frac{1}{2+\delta}}\right] + \frac{1}{2}\sum_{n=1}^{\infty}\left[\frac{C_{2+\delta}2^{2+\delta}\EE\left[|f(X_2)|^{2+\delta}\right]}{2^{(n-1)(2+\delta)/2}}\right]^{\frac{1}{2+\delta}}\\
    &= 2L^{2+\delta}\left[1 + \sup_{i=1,2} \left[\EE\left[\|X_i\|^{2+\delta}\right]\right]^{\frac{1}{2+\delta}}\right] + C_{2+\delta}^{\frac{1}{2+\delta}}L2^{\frac{1+\delta}{2+\delta}}\left[1 + \EE\|X_2\|^{2+\delta}\right]^{\frac{1}{2+\delta}} \sum_{n=1}^{\infty} \frac{\sqrt{2}}{2^{\frac{n}{2}}}\\
    & < \infty. 
\end{align*}

Next, we show that $Y$ satisfies the properties (2) and (3) in the Theorem \ref{thm.one_step_unbiased}. Namely, finite expected sampling complexity and bounded $2+\delta/10$ moment. In order to bound the $2+\delta/10$ moment of $Y$, we introduce the following events:
\begin{align*}
    &E_1 := \left\lbrace \left|\EE [f(X_2) \mid X_1(1)] - f(X_1(1))\right| < \varepsilon\right\rbrace, \\
    &E_2 := \left\lbrace \left|S_{2^{n-1}}^O/2^{n-1} - \EE [f(X_2) \mid X_1(1)]\right|\geq \varepsilon / 2\right\rbrace, \\
    &E_3 := \left\lbrace \left|S_{2^{n-1}}^O/2^{n-1} - S_{2^{n-1}}^E/2^{n-1}\right|\geq \varepsilon / 2\right\rbrace.
\end{align*}
Observe that 
\begin{equation*}
    \EE \left[|\Delta_n|^{2+\delta/10}\right] = \EE \left[|\Delta_n|^{2+\delta/10}\mathbbm{1}(E_1^c\cap E_2^c\cap E_3^c)\right] + \EE \left[|\Delta_n|^{2+\delta/10}\mathbbm{1}(E_1\cup  E_2\cup E_3)\right]. 
\end{equation*}
On the event $E_1^c\cap E_2^c\cap E_3^c$, we have 
\begin{align*}
    \left|\EE [f(X_2) \mid X_1(1)] - f(X_1(1))\right| &\geq \varepsilon,\\
    \left|S_{2^{n-1}}^O/2^{n-1} - \EE [f(X_2) \mid X_1(1)]\right|&\leq \varepsilon / 2,\\
    \left|S_{2^{n-1}}^O/2^{n-1} - S_{2^{n-1}}^E/2^{n-1}\right|&\leq \varepsilon / 2.
\end{align*}
Thus, both $S_{2^{n-1}}^O/2^{n-1}$ and  $S_{2^{n-1}}^E/2^{n-1}$ are on the same side of $f(X_1(1))$. Since $S_{2^n} = S_{2^{n-1}}^O + S_{2^{n-1}}^E$, we get
\begin{equation*}
    \Delta_n = \max\left\lbrace f(X_1(1)), \frac{S_{2^n}}{2^n}\right\rbrace - \frac{1}{2}\left[\max\left\lbrace f(X_1(1)), \frac{S_{2^{n-1}}^O}{2^{n-1}}\right\rbrace + \max\left\lbrace f(X_1(1)), \frac{S_{2^{n-1}}^E}{2^{n-1}}\right\rbrace\right] = 0. 
\end{equation*}
In other words, 
\begin{equation}\label{eq.delta_good}
    \EE \left[|\Delta_n|^{2+\delta/10}\mathbbm{1}(E_1^c\cap E_2^c\cap E_3^c)\right] = 0.
\end{equation}
Next, we bound the term $\EE \left[|\Delta_n|^{2+\delta/10}\mathbbm{1}(E_1\cup  E_2\cup E_3)\right]$. By H\"older's inequality (with parameter $p = (2+\delta) / (2 + \delta / 10)$, and $q = (20 + 10\delta) / (9\delta)$. It is straight forward to verify that $1/p + 1/q = 1$),
\begin{align*}
    \EE \left[|\Delta_n|^{2+\delta/10}\mathbbm{1}(E_1\cup  E_2\cup E_3)\right] &\leq \left[\EE \left[|\Delta_n|^{(2+\delta/10)\cdot \frac{2+\delta}{2+\delta/10}}\right]\right]^{\frac{2+\delta/10}{2+\delta}}\cdot\EE[\mathbbm{1}(E_1\cup  E_2\cup E_3)]^{\frac{9\delta}{20+10\delta}}\\
    &\leq \left[\EE \left[|\Delta_n|^{2+\delta}\right]\right]^{\frac{2+\delta/10}{2+\delta}}\cdot \left(\PP(E_1) + \PP(E_2) + \PP(E_3)\right)^{\frac{9\delta}{20+10\delta}}. 
\end{align*}
Take $\varepsilon = \frac{1}{2^{n/4}}$, by the assumption in \eqref{aspn.bounded_density},
\begin{equation*}
    \PP(E_1) \leq C\varepsilon.
\end{equation*}
Now, we bound the probabilities $\PP(E_2)$ and $\PP(E_3)$. By Corollary \ref{cor.conditional_M-Z_inequality}, there exists a universal constant $C_{2+\delta} >0$, such that 
\begin{align*}
    \PP(E_2)&\leq \frac{1}{(\varepsilon/2)^{2+\delta}}\EE \left[\left|S_{2^{n-1}}^O - \EE [f(X_2) \mid X_1(1)]\right|^{2+\delta}\right]\\
    &\leq \frac{1}{(\varepsilon/2)^{2+\delta}} \cdot \frac{C_{2+\delta}2^{1+\delta}L^{2+\delta}\left[1 + \EE\left[\|X_2\|^{2+\delta}\right]\right]}{2^{(n-1)(2+\delta)/2}}\\
    &= C_{2+\delta}2^{2+3\delta/2}L^{2+\delta}\cdot \left[1 + \EE\left[\|X_2\|^{2+\delta}\right]\right] \varepsilon\cdot \frac{1}{2^{(\frac{1}{4} + \frac{\delta}{4})n - (2+\delta)}}\\
    &\leq C_{2+\delta}2^{2+3\delta/2}L^{2+\delta}\cdot \left[1 + \EE\left[\|X_2\|^{2+\delta}\right]\right]\varepsilon.
\end{align*}

Similarly, 
\begin{equation*}
    \PP (E_3) \leq C_{2+\delta} 2^{4+5\delta/2}L^{2+\delta} \EE \left[1 + \EE \left[\|X_2\|^{2+\delta}\right]\right]\varepsilon.
\end{equation*}
Moreover, recall that 
\begin{align*}
    |\Delta_n|\leq \frac{1}{2}\left|S_{2^{n-1}}^O/2^{n-1} - S_{2^{n-1}}^E/2^{n-1}\right|,
\end{align*}
by Corollary \ref{cor.conditional_M-Z_inequality} again, we have 
\begin{align*}
    \left[\EE \left[ |\Delta_n|^{2+\delta} \right]\right]^{\frac{2+\delta/10}{2+\delta}} &\leq \left(C_{2+\delta}L^{2+\delta}\left[1+ \EE\left[\|X_2\|^{2+\delta}\right]\right]\right)^{\frac{2+\delta/10}{2+\delta}}\cdot\frac{1}{(2^{n})^{\frac{2+\delta/10}{2}}}\\
    &\leq \left(C_{2+\delta}L^{2+\delta}\left[1+ \EE\left[\|X_2\|^{2+\delta}\right]\right]\right)^{\frac{2+\delta/10}{2+\delta}}\cdot\frac{1}{2^n}.
\end{align*}
Thus, there exists universal constant $C'>0$, such that 
\begin{align}\label{eq.delta_bad}
\nonumber\EE \left[|\Delta_n|^{2+\delta/10}\mathbbm{1}(E_1\cup  E_2\cup E_3)\right] &\leq C'\left(L^{2+\delta}\left[1+ \EE\left[\|X_2\|^{2+\delta}\right]\right]\right)^{\frac{2+\delta/10}{2+\delta} + \frac{9\delta}{20+10\delta}} \cdot \frac{1}{2^{\left(1 + \frac{9\delta}{80+40\delta}\right)n}}\\
& \leq C' L^{2+\delta}\left[1+ \EE\left[\|X_2\|^{2+\delta}\right]\right]\frac{1}{2^{\left(1 + \frac{9\delta}{80+40\delta}\right)n}}.
\end{align}
Now, combining \eqref{eq.delta_good} and \eqref{eq.delta_bad}, 
\begin{align*}
    \EE \left[|Y|^{2+\delta/10}\right] & \leq 2^{1+\delta/10} \sum_{n=0}^\infty \frac{\EE\left[|\Delta_n|^{2+\delta/10}\right]}{p_n^{1+\delta/10}} \\
    &\leq \frac{C'}{r} L^{2+\delta}\left[1+ \EE\left[\|X_2\|^{2+\delta}\right]\right]\sum_{n=0}^{\infty}\frac{1}{2^{\left(1 + \frac{9\delta}{80+40\delta}\right)n}2^{-\frac{2+9\delta/(80+40\delta)}{2+\delta/10}n\cdot (1+\delta/10)}}\\
    &= \frac{C'}{r} L^{2+\delta}\left[1+ \EE\left[\|X_2\|^{2+\delta}\right]\right] \sum_{n=0}^{\infty}\frac{1}{\left(2^{\frac{9\delta/(80+40\delta) - \delta/10}{2+\delta/10}}\right)^n} \\
    &\leq \widetilde{C} L^{2+\delta}\left[1+ \EE\left[\|X_2\|^{2+\delta}\right]\right] < \infty,
\end{align*}
where $\widetilde{C}$ is an universal constant independent of the process $(X_1, X_2)$, and we have also used the fact that  $\frac{9\delta}{80+40\delta}-\frac{\delta}{10} > 0$ when $\delta < 1/4$. 
Finally, the sampling complexity of $Y$ is 
\begin{equation*}
    \EE \left[2^N\right] = \sum_{n=0}^\infty 2^np_n \lesssim\sum_{n=0}^{\infty}\frac{1}{\left(2^{\frac{9\delta/(80+40\delta) - \delta/10}{2+\delta/10}}\right)^n} < \infty.
\end{equation*}
Thus, the expected computational cost of $Y$ is also finite. In sum, our estimator $Y$ satisfies all the desired properties in Theorem \ref{thm.one_step_unbiased}.
\end{proof}
Next, we provide the full proof of Theorem \ref{thm.multi_stage_unbiased}. 
\begin{proof}[Proof of Theorem \ref{thm.multi_stage_unbiased}]
By the standard dynamical programming for optimal stopping, we have
\begin{equation*}
\begin{cases}
U_1(X_{1:T-1})= \EE \left[f(X_T) \mid X_{1:T-1}\right],  \\
U_{T-k}(X_{1:k}) = \EE \left[\max \left\lbrace f\left(X_{k+1}\right), U_{T-(k+1)}(X_{1:k+1})\right\rbrace \mid X_{1:k}\right], \quad 0\leq k\leq T-2.
\end{cases}
\end{equation*}
Let $Y_{T-k}(x_{1:k})$ denote the output (which is a random variable) of Algorithm \ref{alg.multi_stage_MLMC} given the input history $x_{1:k}$ which is sampled from $X_{1:k}$. 
For simplicity, let $\delta_k := \delta\cdot 10^{k+1-T}$ for $0 \leq k \leq T-1$. We will prove by a backward induction to show that:
\begin{enumerate}[leftmargin = 0.6cm]
    \item [(a)] $\EE_{\pi_{k+1:T}} [Y_{T-k}(x_{1:k})] = U_{T-k}(x_{1:k})$, \quad $0\leq k \leq T-1$.
    \item [(b)] The expected sampling complexity $=\prod\limits_{i=k+1}^{T-1}C_i < \infty$, $ 0\leq k\leq T-2 $. As a result, the expected computational complexity is also finite. 
    \item [(c)] $\EE_{\pi_{1:T}} \left[\left|Y_{T-k}(x_{1:k})\right|^{2+\delta_k}\right] < \left(\prod\limits_{i=k+1}^{T-1}\widetilde{C}_i\right)\cdot L^{2+\delta}\left[1+ \EE\left[\|X_T\|^{2+\delta}\right]\right]$, \quad {\textrm{for all }} $0\leq k\leq T-2$.
\end{enumerate}
Here $C_i, \widetilde{C}_i$ ($1\leq i\leq T-1$) are some positive constants independent of the underlying process. 

When $k=T-1$, we have $Y_{T-k}(x_{1:T-1}) = f(X_T)$ with $X_T$ sampled from $\pi_{T}(\cdot \mid \{x_i\}_{i=1}^{T-1})$, thus $(a)$ holds by definition. When $k=T-2$, we have $(a), (b)$ and $(c)$ are guaranteed exactly by Theorem \ref{thm.one_step_unbiased}.

Suppose that (a), (b) and (c) are held for $k+1$, where $0\leq k\leq T-3$. Conditioning on the input history $x_{1:k}$ (sampled from $X_{1:k}$), let's sample $x_{k+1}$ from $\pi_{k+1}\left(\cdot \mid \{x_i\}_{i=1}^k \right)$. Then, we sample $N_{k+1}\sim \text{Geo}(r_{k+1})$, and get $i.i.d.$ 
\begin{equation*}
Y_{T-(k+1)}(x_{1:k+1})(1),\cdots, Y_{T-(k+1)}(x_{1:k+1})\left(2^{N_{k+1}}\right).
\end{equation*}
Adapting the same notations as before, we define
\begin{align*}
    S_{2^{N_{k+1}}} &= \sum_{i=1}^{2^{N_{k+1}}}Y_{T-(k+1)}\left(x_{1:k+1}\right)(i),\\
    S^O_{2^{N_{k+1}-1}} &= \sum_{i=1}^{2^{N_{k+1}-1}}Y_{T-(k+1)}\left(x_{1:k+1}\right)(2i-1),\\
    S^E_{2^{N_{k+1}-1}} &= \sum_{i=1}^{2^{N_{k+1}-1}}Y_{T-(k+1)}\left(x_{1:k+1}\right)(2i).
\end{align*}
Then,
$Y_{T-k}(x_{1:k}) = \Delta_{N_{k+1}} / p_{r_{k+1}}(N_{k+1})$, where $\Delta_{N_{k+1}}$ is defined in Algorithm \ref{alg.multi_stage_MLMC}.
Note that by the induction hypothesis, we have
\begin{equation*}
\EE_{\pi_{k+2:T}} \left[ Y_{T-(k+1)}(x_{1:k+1})\right] = U_{T-(k+1)}(x_{1:k+1}),
\end{equation*}
and 
\begin{equation*}
\EE_{\pi_{1:T}} \left[\left|Y_{T-{k+1}}(x_{1:k+1})\right|^{2+\delta_{k+1}}\right] < \left(\prod_{i=k+2}^{T-1}\widetilde{C}_i\right)\cdot L^{2+\delta}\left[1+ \EE\left[\|X_T\|^{2+\delta}\right]\right].
\end{equation*}
We first show that $Y_{T-k}(x_{1:k})$ is an unbiased estimator of $U_{T-k}(x_{1:k})$. 
\begin{align*}
    & \EE_{\pi_{k+1:T}}\left[Y_{T-k}(x_{1:k})\right] \\
    &= \sum_{n=0}^{\infty} \EE_{\pi_{k+1:T}} [\Delta_n] \\
    &=  \sum_{n=1}^{\infty} \left(\EE_{\pi_{k+1:T}}\left[\max\left\lbrace f\left(x_{k+1}\right), \frac{S_{2^n}}{2^n}\right\rbrace\right] - \EE_{\pi_{k+1:T}} \left[\max\left\lbrace f\left(x_{k+1}\right), \frac{S_{2^{n-1}}}{2^{n-1}}\right\rbrace\right] \right) + \\
    &\qquad \EE_{\pi_{k+1:T}} \left[\max\left\lbrace f\left(x_{k+1}\right), Y_{T-(k+1)}(x_{1:k+1})(1)\right\rbrace\right]\\
    &= \EE_{\pi_{k+1:T}} \left[\max\left\lbrace f\left(x_{k+1}\right), \lim_{n\rightarrow \infty}\frac{S_{2^n}}{2^n}\right\rbrace\right]\\
    &= \EE_{\pi_{k+1:T}} \left[\max\left\lbrace f\left(x_{k+1}\right), U_{T-(k+1)}(x_{1:k+1})]\right\rbrace\right] = U_{T-k}(x_{1:k}). 
\end{align*}
Next, we bound the expected value $\EE_{\pi_{1:T}} \left[\left|Y_{T-k}(x_{1:k})\right|^{2+\delta_k}\right]$. For simplicity, in the following proof, we use $\PP$ and $\EE$ as abbreviations of $\PP_{\pi_{1:T}}$ and $\EE_{\pi_{1:T}}$.

Following the same idea in the proof of Theorem \ref{thm.one_step_unbiased}, we introduce three events:
\begin{align*}
    &E_1' := \left\lbrace \left|U_{T-(k+1)}(x_{1:k+1}) - f(x_{k+1})\right| < \varepsilon\right\rbrace, \\
    &E_2' := \left\lbrace \left|S_{2^{n-1}}^O/2^{n-1} - U_{T-(k+1)}(x_{1:k+1})\right|\geq \varepsilon / 2\right\rbrace, \\
    &E_3' := \left\lbrace \left|S_{2^{n-1}}^O/2^{n-1} - S_{2^{n-1}}^E/2^{n-1}\right|\geq \varepsilon / 2\right\rbrace.
\end{align*}
We start with bounding the probability of event $E_1'$. Take $\varepsilon = \frac{1}{2^{n/4}}$ (the same as Theorem \ref{thm.one_step_unbiased}), by Assumption \ref{ass.conditional}, we have
\begin{equation*}
\PP(E_1') = \int \mathbbm{1}\left(\left\lbrace \left|U_{T-(k+1)}(x_{1:k+1}) - f(x_{k+1})\right| < \varepsilon\right\rbrace\right)\pi_{1:k+1}(dx_1,\cdots, dx_{k+1})\leq C\varepsilon.
\end{equation*}

Similar to the proof of Theorem \ref{thm.one_step_unbiased}, by conditioning on $x_{1:k}$, we can apply Lemma \ref{lemma.M-Z_inequality} to get
\begin{align*}
    \PP_{\pi_{k+1:T}}(E_2') &\leq \frac{1}{(\varepsilon/2)^{2+\delta_{k+1}}} \cdot \frac{C_{2+\delta_{k+1}}2^{1+\delta_{k+1}}\left[1 + \EE\left[\left|Y_{T-(k+1)}\right|^{2+\delta_{k+1}} \mid x_{1:k}\right]\right]}{2^{(n-1)(2+\delta_{k+1})/2}}\\
    &\leq C_{2+\delta_{k+1}}2^{2+3\delta_{k+1}/2}\left[1 + \EE\left[\left|Y_{T-(k+1)}(x_{1:k+1})\right|^{2+\delta_{k+1}} \mid x_{1:k}\right]\right]\varepsilon.
\end{align*}
Thus, we can bound the probability of event $E_2'$ under $\pi_{1:T}$ by
\begin{align*}
    \PP(E_2') &= \int \PP_{\pi_{k+1:T}}(E_2') \pi_{1:k}(dx_{1:k})\\
    &\leq C_{2+\delta_{k+1}}2^{2+3\delta_{k+1}/2}\left[1 + \int\EE\left[\left|Y_{T-(k+1)}(x_{1:k+1})\right|^{2+\delta_{k+1}} \mid x_{1:k}\right]\pi_{1:k}(dx_{1:k})\right]\varepsilon\\
    &= C_{2+\delta_{k+1}}2^{2+3\delta_{k+1}/2}\left[1 + \EE\left[\left|Y_{T-(k+1)}(x_{1:k+1})\right|^{2+\delta_{k+1}}\right]\right]\varepsilon.
\end{align*}
Similarly, 
\begin{equation*}
    \PP(E_3') \leq C_{2+\delta_{k+1}}2^{4+5\delta_{k+1}/2}\left[1 + \EE\left[\left|Y_{T-(k+1)}(x_{1:k+1})\right|^{2+\delta_{k+1}}\right]\right]\varepsilon. 
\end{equation*}
Moreover, 
\begin{align*}
    \left[\EE \left[ |\Delta_n|^{2+\delta_{k+1}} \right]\right]^{\frac{2+\delta_k}{2+\delta_{k+1}}}  &\leq \left(C_{2+\delta_{k+1}}\left[1+ \EE\left[\left|Y_{T-(k+1)}(x_{1:k+1})\right|^{2+\delta_{k+1}}\right]\right]\right)^{\frac{2+\delta_k}{2+\delta_{k+1}}}\cdot\frac{1}{(2^{n})^{\frac{2+\delta_{k}}{2}}}\\
    &\leq \left(C_{2+\delta_{k+1}}\left[1+ \EE\left[\left|Y_{T-(k+1)}(x_{1:k+1})\right|^{2+\delta_{k+1}}\right]\right]\right)^{\frac{2+\delta_k}{2+\delta_{k+1}}}\cdot\frac{1}{2^n}.
\end{align*}
Following the same technique in the proof of Theorem \ref{thm.one_step_unbiased}, there exists universal constant $C'>0$ independent of the underlying process and $T$, such that 
\begin{align*}
&\EE \left[|\Delta_n|^{2+\delta_k}\mathbbm{1}(E_1'\cup  E_2'\cup E_3')\right] \\
\leq\quad& C'\left(\left[1+ \EE\left[\left|Y_{T-(k+1)}(x_{1:k+1})\right|^{2+\delta_{k+1}}\right]\right]\right)^{\frac{2+\delta_k}{2+\delta_{k+1}} + \frac{9\delta_{k+1}}{20+10\delta_{k+1}}} \cdot \frac{1}{2^{\left(1 + \frac{9\delta_{k+1}}{80+40\delta_{k+1}}\right)n}}\\
\leq\quad& C'\left[1+ \EE\left[\left|Y_{T-(k+1)}(x_{1:k+1})\right|^{2+\delta_{k+1}}\right]\right]\frac{1}{2^{\left(1 + \frac{9\delta_{k+1}}{80+40\delta_{k+1}}\right)n}}.
\end{align*}
Here we have used the fact that $\frac{2+\delta_k}{2+\delta_{k+1}} + \frac{9\delta_{k+1}}{20+10\delta_{k+1}} = 1$. Noticing that 
\begin{equation*}
 \EE \left[|\Delta_n|^{2+\delta_k}\mathbbm{1}(E_1'^c\cap  E_2'^c\cap E_3'^c)\right] = 0,
\end{equation*}
we get
\begin{align*}
    &\EE \left[\left|Y_{T-k}(x_{1:k})\right|^{2+\delta_k}\right] \\
    \leq \quad & 2^{1+\delta_k} \sum_{n=0}^\infty \frac{\EE\left[|\Delta_n|^{2+\delta_k}\right]}{\left[p_{r_{k+1}}(n)\right]^{1+\delta_k}} \\
    = \quad & \frac{C'}{r_{k+1}} \left[1+ \EE\left[\left|Y_{T-(k+1)}(x_{1:k+1})\right|^{2+\delta_{k+1}}\right]\right] \sum_{n=0}^{\infty}\frac{1}{\left(2^{\frac{9\delta_{k+1}/(80+40\delta_{k+1}) - \delta_k}{2+\delta_k}}\right)^n} \\
    \leq \quad & \widetilde{C}_{k+1}\left(\prod_{i=k+2}^{T-1}\widetilde{C}_i\right) L^{2+\delta}\left[1+ \EE\left[\|X_T\|^{2+\delta}\right]\right] = \left(\prod_{i=k+1}^{T-1}\widetilde{C}_i\right) L^{2+\delta}\left[1+ \EE\left[\|X_T\|^{2+\delta}\right]\right],
\end{align*}
where in the last inequality we have applied the induction hypothesis $(c)$. Note that $9\delta_{k+1}/(80+40\delta_{k+1}) - \delta_k > 0$ when $\delta < 1/4$, we have 
\begin{align}
\nonumber\widetilde{C}_{k+1}:&= \frac{C'}{r_{k+1}}\sum_{n=0}^{\infty}\frac{1}{\left(2^{\frac{9\delta_{k+1}/(80+40\delta_{k+1}) - \delta_k}{2+\delta_k}}\right)^n}\\
\nonumber&= \frac{C'}{r_{k+1}}\cdot\frac{1}{2^{\frac{9\delta_{k+1}/(80+40\delta_{k+1}) - \delta_k}{2+\delta_k}} - 1}\\
\nonumber&\leq \frac{C'}{r_{k+1}}\cdot \frac{1}{\frac{9\delta_{k+1}/(80+40\delta_{k+1}) - \delta_k}{2+\delta_k}} \quad (\textrm{note that }2^{\alpha}\geq 1 + \alpha \textrm{ for }\alpha >0 )\\
\nonumber&= \frac{C'}{r_{k+1}}\cdot \frac{(2 + \delta_{k+1}/10)(80 + 40\delta_{k+1})}{9\delta_{k+1} - (8+4\delta_{k+1})\delta_{k+1}} \quad (\delta_{k} = \delta_{k+1} / 10)\\
\nonumber&\leq \frac{C'}{r_{k+1}}\cdot \frac{3 \cdot 90}{1-4\delta}\cdot \frac{1}{\delta_{k+1}} \quad (\delta_{k+1} < \delta < 1/4)\\
&= \frac{C'}{r_{k+1}}\cdot \frac{3}{\delta(1-4\delta)}\cdot 10^{T-k} \quad (\delta_{k+1} = \delta\cdot10^{k+2-T}) \label{bound_tilde_C}\\
\nonumber& < \infty
\end{align}
is a constant independent of underlying process.
Finally, since we have called Algorithm \ref{alg.multi_stage_MLMC} for $2^{N_{k+1}}$ times to construct $Y_{T-k}(x_{1:k})(i)$ $(1\leq i\leq 2^{N_{k+1}})$, the expected sampling complexity of computing $Y_{T-k}(x_{1:k})$ is 
\begin{equation*}
    \EE[2^{N_{k+1}}]\cdot \prod_{i=k+2}^{T-1}C_i = \prod_{i=k+1}^{T-1} C_i < \infty,
\end{equation*}
where
\begin{align}
C_{k+1}:= \EE[2^{N_{k+1}}] \lesssim \sum_{n=0}^\infty \frac{1}{\left(2^{\frac{9\delta_{k+1}/(80+40\delta_{k+1}) - \delta_k}{2+\delta_k}}\right)^n} \leq \frac{3}{\delta(1-4\delta)}\cdot 10^{T-k}. \label{bound_C}
\end{align}
As a result, the expected computational complexity is also finite. To sum up, $(a), (b)$ and $(c)$ are satisfied for $k$. Thus, the proof by induction is completed. In particular, together with \eqref{bound_tilde_C} and \eqref{bound_C}, there exist universal constant $D>0$ independent of the underlying process and $T$, such that the resulting estimator $Y_T$ in Algorithm \ref{alg.multi_stage_MLMC} satisfying:
\begin{enumerate}
    \item [(1)] $\EE[Y_T] = U_T$
    \item [(2)] Expected computational complexity is $\prod_{i=1}^T C_i < D\cdot 10^{T^2}$ (by \eqref{bound_C}).
    \item [(3)] The variance of $Y_T$ is bounded by
    $$
    \left(\prod\limits_{i=k+1}^{T-1}\widetilde{C}_i\right)\cdot L^{2+\delta}\left[1+ \EE\left[\|X_T\|^{2+\delta}\right]\right] < D\cdot 10^{T^2}\cdot L^{2+\delta}\left[1+ \EE\left[\|X_T\|^{2+\delta}\right]\right].
    $$
\end{enumerate}
\end{proof}